\newif\iffullmode
\fullmodetrue 
\iffullmode
  \documentclass[acmsmall,nonacm]{acm/acmart}
\else
  \documentclass[acmsmall]{acm/acmart}
\fi

\usepackage{graphicx}
\usepackage{amsmath}
\usepackage{algorithm}
\usepackage[noend]{algpseudocode}
\usepackage{tikz}
\usepackage{xcolor}
\usetikzlibrary{decorations.pathreplacing}
\usetikzlibrary{arrows.meta}
\usetikzlibrary{calc}

\definecolor{pRed}{HTML}{C4361D}
\definecolor{pBlue}{HTML}{1F6FB2}
\definecolor{pTeal}{HTML}{0F8A7E}
\definecolor{pOrange}{HTML}{D8691C}
\definecolor{pPurple}{HTML}{9B3B8C}
\definecolor{pGray}{HTML}{5C5C5C}

\usepackage{diagbox}
\usepackage{booktabs}
\usepackage{colortbl}
\usepackage{thmtools}

\newcommand{\bx}[0]{\mathbf{x}}

\newcommand{\mL}{\mathcal{L}}
\newcommand{\mC}{\mathcal{C}}
\newcommand{\mH}{\mathcal{H}}    
\newcommand{\smH}{\dot{\mathcal{H}}}
\newcommand{\mT}{\mathcal{T}}
\newcommand{\mA}{\mathcal{A}}

\newcommand{\mE}{\mathcal{E}}

\newcommand{\bU}{\mathbf{U}}

\newcommand{\sfhow}{\textsf{sfhow}}
\newcommand{\cqneg}{\texttt{CQ}^{\neg}}
\newcommand{\cq}{\texttt{CQ}}

\newcommand{\apporfull}[1]{\iffullmode Appendix~\ref{#1}\else the full paper~\cite{FullVersion}\fi}

\newcommand{\rev}[1]{#1}    

\AtBeginDocument{} 

\begin{document}

\title{Conjunctive Queries with Negation: Beyond Signed-Acyclicity}

\author{Simon Frisk}
\affiliation{
  \institution{University of Wisconsin--Madison}
  \city{Madison}
  \state{Wisconsin}
  \country{USA}
}
\email{simon.frisk@wisc.edu}

\author{Paraschos Koutris}
\affiliation{
  \institution{University of Wisconsin--Madison}
  \city{Madison}
  \state{Wisconsin}
  \country{USA}
}
\email{paris@cs.wisc.edu}

\begin{abstract}
We study constant-delay enumeration for conjunctive queries with negation ($\cqneg$). \rev{It is known that \emph{signed-acyclic} queries can be enumerated with constant delay after linear preprocessing time, using linked lists with skip links.}
We \rev{use} a hypergraph width measure for $\cqneg$, the \rev{\emph{fractional signed hyperorder width}}
($\sfhow$), \rev{which requires} a single variable order \rev{to} simultaneously handle every subset
of the negative atoms\rev{, which Capelli et al.\ introduced to characterize the preprocessing time of direct access}.
$\sfhow$ strictly generalizes signed-acyclicity (recovered when $\sfhow = 1$) and
fractional hypertree width (recovered on queries without negation). Our main algorithmic result is a variable-elimination algorithm that achieves constant-delay enumeration on input $I$ for any full $\cqneg$ query with preprocessing time $O(|I|^{\sfhow})$.
Furthermore, we show that using multiple variable orders can improve this bound, exhibiting
an $O(|I|^{3/2})$ algorithm for $k$-cycle queries with at least one negative edge.

\end{abstract}

\maketitle

\section{Introduction}

Conjunctive Queries with Negation ($\cqneg$) form a natural class of database queries. Every SQL query built from joins, projections, and antijoins (\texttt{NOT IN}, \texttt{NOT EXISTS}, \texttt{EXCEPT}) is a $\cqneg$ query. Furthermore, negation is the standard mechanism for expressing set difference, integrity-constraint violations, and graph patterns that forbid certain substructures (e.g., open triangles, or paths that avoid a designated set of edges). Despite their ubiquity, $\cqneg$ queries are algorithmically far less understood than their positive counterparts.

In this paper, we study the problem of enumerating the output tuples of a $\cqneg$ query. A $\cqneg$ query extends a standard Conjunctive Query ($\cq$) by allowing negated atoms in the query body. Formally, it is of the form
\begin{align}
Q(\bx_F) \leftarrow \bigwedge_{K \in \mE^+} R_K(\bx_K) \wedge
\bigwedge_{N \in \mE^-} \neg R_N(\bx_N)
\end{align}
where the variables are of the form $x_i$, $i \in [n] = \{1, 2, \dots, n\}$, and $\mE^+, \mE^-$ are two multisets of hyperedges, called {\em positive} and {\em negative}, respectively. Each hyperedge is a subset of $[n]$. The variables $\{x_i \mid i \in F \}$, for $F \subseteq [n]$, are the {\em free variables}. When $F = [n]$, we say that $Q$ is a full query. We focus on full queries in this paper. The structure $\smH = ([n], \mE^+, \mE^-)$ associated with $Q$ is called the {\em signed hypergraph} of $Q$. We use the dotted notation to distinguish it from a standard hypergraph.

Given a query $Q$ and a database instance $I$, the {\em enumeration task} asks to first construct a data structure during a preprocessing phase, after which the answers $Q(I)$ can be enumerated with as little delay as possible. The ideal goal is linear preprocessing time $O(|I|)$ with constant-delay enumeration. In this work, we will focus on algorithms for $\cqneg$ queries that achieve constant delay, albeit with superlinear preprocessing time.

The computational complexity of enumeration is well-understood for several subclasses of $\cqneg$. When $Q$ only has positive hyperedges ($\mE^- = \emptyset$), the query becomes a standard $\cq$ with hypergraph $\mH = ([n], \mE^+)$. We will sometimes call this a \emph{positive} query. In this case, a combination of a constant-delay enumeration algorithm for acyclic queries~\cite{BaganDG07} with any worst-case optimal join~\cite{NPRR,LeapfrogTriejoin} gives a constant-delay algorithm with preprocessing time $O(|I|^{\textsf{fhw}(\mH)})$ for full queries, where $\textsf{fhw}(\mH)$ is the fractional hypertree width of $\mH$. This algorithm performs the enumeration task following only one tree decomposition of $\mH$. The PANDA algorithm~\cite{Panda,PandaExpress} instead uses multiple tree decompositions to improve the preprocessing time to $O(|I|^{\textsf{subw}(\mH)})$, where $\textsf{subw}(\mH)$ is the submodular width of $\mH$. This is the state-of-the-art preprocessing-time bound for positive full queries.

When $\smH$ contains a mixture of positive and negative hyperedges, the landscape is much less understood. We do know the subclass of queries for which we need only linear-time preprocessing to achieve constant delay~\cite{SignedAcyclic}: these are the free-connex {\em signed-acyclic} queries~\cite{BraultBaronThesis}. Signed-acyclicity generalizes $\alpha$-acyclicity by requiring that for every subset $S \subseteq \mE^-$, the hypergraph $\mH^S = ([n],\mE^+ \cup S)$ is $\alpha$-acyclic. Beyond the linear-time regime, not much is known about \rev{constant-delay} enumeration, although the decision problem of testing whether $Q(I)$ is non-empty has been studied somewhat more extensively~\cite{BraultBaron12,SignedAcyclic,NestSetWidth}. 

\rev{Capelli et al.~\cite{CapelliIrwin24,CapelliCarmeliIrwinSalvati26} gave two elegant direct-access algorithms for $\cqneg$ queries: one reduces direct access to positive queries, and the other builds on constructing an ordered circuit that represents the answers.}
\rev{\rev{The preprocessing time of their algorithms is governed by the \emph{fractional signed hyperorder width} ($\sfhow$)~\cite{CapelliIrwin24,CapelliCarmeliIrwinSalvati26}.} Rather than measuring the width of $\smH$ directly, it is measured on all \emph{positive completions} $\mH^S = ([n], \mE^+ \cup S)$, and it is required that all of these $2^{|\mE^-|}$ completions admit a \emph{common} variable order. Formally,
\[
\sfhow(\smH) \;:=\; \min_{\sigma}\; \max_{S \subseteq \mE^-}\; \max_{i \in [n]}\; \rho^*_{\mH^S}\bigl(\chi_\sigma^S(x_i)\bigr),
\]
where $\sigma$ is a variable order, $\chi_\sigma^S(x_i)$ is the bag of variable $x_i$ in the canonical tree decomposition of $\mH^S$ induced by $\sigma$, and $\rho^*$ is the fractional edge cover number. The crucial design choice is the order of quantifiers: the variable order $\sigma$ is chosen \emph{once} and must serve every $S$ simultaneously, rather than being tailored per~$S$. This min-max formulation is what allows a single variable-elimination algorithm to handle the exponentially many positive completions at once.}
\rev{For full queries, Capelli et al. obtain preprocessing time ${O}(|I|^{\sfhow})$ and polylogarithmic access time, and hence enumerate $Q(I)$ with polylogarithmic delay.}

\rev{The width $\sfhow$ is a strict generalization of signed-acyclicity, in the sense that $\sfhow=1$ if and only if the query is signed-acyclic. Also, $\sfhow$ is a strict generalization of fractional hypertree width, meaning that $\sfhow$ coincides with fractional hypertree width on positive queries, but there are queries with negation where $\sfhow$ is strictly larger than fractional hypertree width~\cite{CapelliIrwin24,CapelliCarmeliIrwinSalvati26}.}

\rev{The constant-delay algorithm of~\cite{SignedAcyclic} for signed-acyclic queries, i.e., queries with $\sfhow = 1$, is based on a different technique: it eliminates the variables one at a time and stores, for each variable, linked lists of candidate values with \emph{skip links} that jump over values blocked by negative atoms. Thus, there are two separate lines of work: direct access handles arbitrary $\sfhow$ but gives only polylogarithmic delay, while the skip-link framework gives constant delay but only for $\sfhow = 1$.}

\subsection{Our Contributions}

Based on the above discussion, we ask the following question in this work:

\begin{quote}
 {\em \rev{Can we achieve constant-delay enumeration for full $\cqneg$ queries with preprocessing time $O(|I|^{\sfhow})$?}}
\end{quote}

\rev{We answer this question positively.} Our algorithm extends the variable-at-a-time enumeration framework of~\cite{SignedAcyclic} to arbitrary $\sfhow$. Fix any variable order $\sigma$ witnessing $\sfhow(\smH)$. The algorithm eliminates the variables in the order $\sigma$, building at each step a data structure that, given an assignment to the variables eliminated so far, enumerates the consistent extensions to the next variable. The use of a single variable order shared across all positive completions is what allows the construction cost at each step to be bounded uniformly by the AGM bound\rev{, similar to the analysis of~\cite{CapelliIrwin24,CapelliCarmeliIrwinSalvati26}}, summing to a total preprocessing cost of $O(|I|^{\sfhow})$. The main technical challenge is to merge, into a single data structure, the obstructions imposed by all subsets $S \subseteq \mE^-$ so that, during the enumeration phase, the next answer can be located after a bounded number of operations.

A natural follow-up question is whether $\sfhow$ is the right bound, or whether it can be tightened in analogy to how fractional hypertree width is tightened to submodular width via the use of multiple tree decompositions. In our setting, the analogous move is to use \emph{multiple} variable orders, each tailored to a different region of the input, rather than committing to a single $\sigma$. We obtain such an improvement for the family of $k$-cycle queries with negation, where a partitioning argument allows us to combine several variable orders to beat the $\sfhow$ bound. Generalizing this multi-variable-order approach to arbitrary $\cqneg$ queries is left as future work.

\section{Preliminaries}

\paragraph{The class $\cqneg$}
A Conjunctive Query with Negation ($\cqneg$) is defined as follows:
\begin{align} \label{eq:cqneg}
Q(\bx_F) \leftarrow \bigwedge_{K \in \mE^+} R_K(\bx_K) \wedge
\bigwedge_{N \in \mE^-} \neg R_N(\bx_N)
\end{align}
where the variables are of the form $x_i$, $i \in [n] = \{1, 2, \dots, n\}$, and $\mE^+, \mE^-$ are two multisets of hyperedges, called {\em positive} and {\em negative}, respectively. Each hyperedge is a subset of $[n]$. The variables $\{x_i \mid i \in F \}$, for $F \subseteq [n]$, are the {\em free variables}. When $F = [n]$, we say that $Q$ is a full query. The structure $\smH = ([n], \mE^+, \mE^-)$ associated with $Q$ will be called the {\em signed hypergraph}
of $Q$. In this work, we will study only {\em safe queries}, where $\bigcup_{K \in \mE^+} K = [n]$, i.e., every variable appears in at least one positive atom.
When $\mE^-  = \emptyset$, then $Q$ is a Conjunctive Query ($\cq$) with the associated hypergraph $([n], \mE^+)$.

\paragraph{Enumeration Complexity} We will study the enumeration problem
$\textsc{Enum}(Q, I)$ for $\cqneg$ queries, where the goal is to output all answers to a given $\cqneg$ query $Q$
on a given database instance $I$.

An enumeration algorithm for $\textsc{Enum}(Q, I)$ consists of two phases:
\begin{itemize}
    \item {\bf Preprocessing phase:} it constructs efficient data structures from $Q$ and $I$; and
    \item {\bf Enumeration phase:} it may access the data structures built during preprocessing, and emit
    the answers of $Q(I)$ one by one, without repetitions.
\end{itemize}

We say that an enumeration algorithm enumerates with delay $O(\tau)$ if the time between the
emission of any two consecutive answers (and the time to emit the first answer, and the time
from the last answer to the end) is bounded by $O(\tau)$. In particular, we say that an
enumeration algorithm is {\em constant-delay} if it enumerates with delay independent of the
input database size $|I|$.

\paragraph{Model of Computation}
In this paper, we use the random-access machine (RAM)
as our computation model with $O(\log |I|)$-bit words.
Our runtime results will be with respect to the {\em data complexity} of evaluation, i.e., the
complexity is measured in the size of the database instance, while the query is fixed.
For a database instance $I$, we define $|I|$ as the sum of the size of all relations in the instance.
\section{Variable Orders, Tree Decompositions, and Widths}
\label{sec:nice-td}

In this section, we will define \rev{the notion of width} we use in the paper to capture the
complexity of evaluating Conjunctive Queries with negation. For reasons that will become clear later,
it will be more convenient to work with variable orders instead of tree decompositions. More background on the connection between variable orders and tree decompositions can be found in~\cite{FAQ}. Some of the lemmas below are from~\cite{FAQ}. Missing proofs of the remaining lemmas are in the appendix.

\subsection{A Refresher on Variable Orders}
\label{sec:varorderrefresh}

We start by defining tree decompositions and the general notion of $f$-width.

\begin{definition}[Tree Decomposition]
A tree decomposition of a hypergraph $\mH= ([n],\mE)$ is a pair $(T,\chi)$ where $T$
is a tree and $\chi : V(T) \to 2^{[n]}$ maps each node $t$ of the tree (called a bag) to a subset $\chi(t)$ of vertices such that:
\begin{enumerate}
    \item For every hyperedge $K \in \mE$, there is a node $t$ of $T$ such that $K \subseteq \chi(t)$.
    \item For every $i \in [n]$, the set $\{t \in V(T) \mid i \in \chi(t)\}$ is a non-empty (connected) subtree of $T$.
\end{enumerate}
\end{definition}

We will denote by $\textsf{TD}(\mH)$ the set of tree decompositions of a hypergraph $\mH$.
Given a set function $f: 2^{[n]} \to \mathbb{R}^+$, we define the $f$-width of a hypergraph $\mH$ as:
\[\text{width}_f(\mH) := \min_{(T,\chi) \in \textsf{TD}(\mH)} \max_{t \in V(T)} f(\chi(t)).\]

\paragraph{Variable Orders}
We now switch to defining variable orders following~\cite{FAQ}. For a hypergraph $\mH=([n],\mE)$, a {\em variable order} $\sigma$ is simply a total order on $[n]$.
Given a vertex $x_i \in [n]$, define $\partial_{\mH} x_i = \{K \in \mE : x_i \in K\}$, i.e., the set of hyperedges of $\mH$ containing $x_i$. The variable order $\sigma = (\sigma_1, \dots, \sigma_n)$ defines inductively a sequence of hypergraphs $\mH_0, \mH_1, \dots, \mH_n$ where $\mH_0 = \mH$ and for $i=1, \dots, n$, $\mH_{i}$ is obtained from $\mH_{i-1}$ as follows:
\begin{enumerate}
    \item Remove the vertex $\sigma_i$.
    \item Remove all hyperedges in $\partial_{\mH_{i-1}} \sigma_i$.
    \item Add the hyperedge $U_i \setminus \{\sigma_i\}$, where $U_i := \bigcup_{K \in \partial_{\mH_{i-1}} \sigma_i} K$.
\end{enumerate}

This construction is analogous to constructing a chordal completion of a graph given an elimination order. We will now associate to $\sigma$ an object $\mT_\sigma = (T_\sigma, \chi_\sigma)$, which we will call the {\em $\sigma$-canonical tree decomposition} (and which we will show below to indeed be a tree decomposition) defined as follows: the nodes of $\mT_\sigma$ are $t_1, \dots, t_n$, where $t_i$ corresponds to the vertex $\sigma_i$. The bag $\chi_\sigma(t_i)$ is defined as $U_i$. Finally, for each $i$, we connect $t_i$ to the node $t_j$ with $j > i$ such that $\chi_\sigma(t_j) \supseteq U_i \setminus \{\sigma_i\}$ and $j$ is minimal with this property.

\begin{lemma}[\cite{FAQ}]
    Let $\mH = ([n],\mE)$ be a hypergraph, and let $\sigma$ be a variable order over $[n]$. Then, $\mT_\sigma$ is a tree decomposition of $\mH$.
\end{lemma}

From the above construction, the $n$ bags of $\mT_\sigma$ correspond exactly to the sets
$\{ U_i \mid i \in [n] \}$. Moreover, $\sigma_i \in U_i$ and the largest $j \in [n]$ such that $\sigma_i \in U_j$ is $j=i$. We thus say that $\sigma_i$ {\em disappears} in the bag $U_i$, and denote by $\chi_\sigma(\sigma_i)$ the bag where $\sigma_i$ disappears, i.e., $\chi_\sigma(\sigma_i) = U_i$.

In order to relate variable orders back to tree decompositions, we need the following lemma.

\begin{lemma}[\cite{FAQ}]
Let $\mH$ be a hypergraph and $\mT=(T,\chi)$ be a tree decomposition of $\mH$. There exists a variable order $\sigma$ such that for every node $v' \in V(\mT_\sigma)$, there is a node $v \in V(\mT)$ such that $\chi_\sigma(v')\subseteq \chi(v)$.
\end{lemma}

Given the above lemma, for any monotone set function $f$, the $f$-width of a hypergraph $\mH$ equals the minimum, over variable orders $\sigma$, of the maximum $f$-value attained over the bags of $\mT_\sigma$.

\begin{restatable}[\cite{FAQ}]{lemma}{transformtocanonicaltdlemma}
\label{lemma:TransformToCanonicalTD}
Let $\mH = ([n],\mE)$ be a hypergraph and $f: 2^{[n]} \to \mathbb{R}^+$ be monotone. Then,
\[\text{width}_f(\mH) = \min_\sigma \max_{t \in V(\mT_\sigma)} f(\chi_\sigma(t))
 = \min_\sigma \max_{i \in [n]} f(\chi_\sigma(\sigma_i)).\]
\end{restatable}

\subsection{Variable Orders for CQs with Negation}

Consider now a $\cqneg$ with signed hypergraph $\smH=([n],\mE^+,\mE^-)$. For any $S\subseteq \mE^-$, we denote by $\mH^S$ the hypergraph $([n],\mE^+\cup S)$.
Let us fix a variable order $\sigma$ on $[n]$. Then, for any $S\subseteq \mE^-$,
we can construct the $\sigma$-canonical tree decomposition $\mT_\sigma^S$ of
the hypergraph $\mH^S$. We will denote by $\chi_\sigma^S(x_i)$ the (unique) bag
in $\mT_\sigma^S$ where vertex $x_i$ disappears. We note here that the tree structure of $\mT_\sigma^S$ for different $S$ is not the same, as the following example shows.

\begin{example}
\label{example:RunningExample}
    Consider the query
    \[Q(x,y,z,w) \gets R(x,y) \wedge S(z,w) \wedge \neg T(x,y,z,w)\]
    and the variable order $\sigma=(x,z,y,w)$. The table below shows the $\sigma$-canonical tree decompositions $\mT_\sigma^\emptyset$ and $\mT_\sigma^{\{T\}}$: each cell is the bag $\chi_\sigma^S(x_i)$ where the column picks the disappearing variable $x_i$ (shown underlined) and the row picks $S$. The arrows within each row indicate the parent edges of the corresponding tree decomposition (each variable's bag points to the bag of its parent in $\mT_\sigma^S$).

    \begin{center}
    \scalebox{0.8}{\begin{tikzpicture}[
    cell/.style={font=\small, inner xsep=3pt, inner ysep=2pt},
    arr/.style={->, >=Latex, semithick, gray!80, shorten <=3pt, shorten >=3pt},
]
\def\rl{1.45}     
\def\redge{9.30}  
\def\xC{2.45}\def\zC{4.65}\def\yC{6.65}\def\wC{8.50}
\def\rowA{-0.74}  
\def\rowB{-1.52}  
\def\hb{-0.36}    
\def\rmid{-1.13}  
\def\bot{-1.90}
\fill[gray!8]    (0,0)    rectangle (\redge,\hb);
\fill[gray!8]    (0,\hb)  rectangle (\rl,\bot);
\fill[pOrange!10] (\rl,\rmid) rectangle (\redge,\bot);
\draw[gray!25] (0,0) rectangle (\redge,\bot);
\draw[gray!25] (0,\hb)--(\redge,\hb);
\draw[gray!25] (0,\rmid)--(\redge,\rmid);
\draw[gray!25] (\rl,0)--(\rl,\bot);
\node[font=\small\bfseries] at (\xC,-0.18) {$x$};
\node[font=\small\bfseries] at (\zC,-0.18) {$z$};
\node[font=\small\bfseries] at (\yC,-0.18) {$y$};
\node[font=\small\bfseries] at (\wC,-0.18) {$w$};
\node[font=\small\bfseries, text=pGray]   at (0.725,\rowA) {$S=\emptyset$};
\node[font=\small\bfseries, text=pOrange] at (0.725,\rowB) {$S=\{T\}$};
\node[cell] (x0) at (\xC,\rowA) {$\{\underline{x},y\}$};
\node[cell] (z0) at (\zC,\rowA) {$\{\underline{z},w\}$};
\node[cell] (y0) at (\yC,\rowA) {$\{\underline{y}\}$};
\node[cell] (w0) at (\wC,\rowA) {$\{\underline{w}\}$};
\node[cell] (x1) at (\xC,\rowB) {$\{\underline{x},y,z,w\}$};
\node[cell] (z1) at (\zC,\rowB) {$\{y,\underline{z},w\}$};
\node[cell] (y1) at (\yC,\rowB) {$\{\underline{y},w\}$};
\node[cell] (w1) at (\wC,\rowB) {$\{\underline{w}\}$};
\draw[arr] (x0) to[bend left=14] (y0);
\draw[arr] (y0) -- (w0);
\draw[arr] (z0) to[bend right=14] (w0);
\draw[arr] (x1) -- (z1);
\draw[arr] (z1) -- (y1);
\draw[arr] (y1) -- (w1);
\end{tikzpicture}}
    \end{center}
\end{example}

\begin{restatable}{lemma}{bagscontaineddifferentslemma}
\label{lemma:BagsContainedDifferentS}
    Let $\mH=([n],\mE^+,\mE^-)$ be a signed hypergraph and let $\sigma$ be a variable order. Let $S_1 \subseteq S_2 \subseteq \mE^-$. Then, for every $i \in [n]$, we have
    $
    \chi_\sigma^{S_1}(x_i) \subseteq \chi_\sigma^{S_2}(x_i)
    $.
\end{restatable}

In particular, for any $S\subseteq \mE^-$, for any $i \in [n]$, we have $\chi_\sigma^\emptyset(x_i) \subseteq \chi_\sigma^S(x_i)$. In Example~\ref{example:RunningExample}, each column of the table satisfies $\chi_\sigma^\emptyset(x_i)\subseteq\chi_\sigma^{\{T\}}(x_i)$, as guaranteed by Lemma~\ref{lemma:BagsContainedDifferentS}.

\subsection{\texorpdfstring{\rev{The}}{The} Width Measure for CQs with Negation}

Given a hypergraph $\mH = ([n],\mE)$ and a subset of variables $X \subseteq [n]$, we define the {\em fractional edge cover number} of $X$ in $\mH$ as
 \[\rho_\mH^*(X) := \min \left\{ \sum_{K \in \mE} w_K \mid \forall K \in \mE: w_K \geq 0, \forall i \in X: \sum_{K : i \in K} w_K \geq 1 \right\}.\]
The fractional hypertree width of the hypergraph $\mH$ is
\[
\textsf{fhw}(\mH) = \min_{(T,\chi) \in \textsf{TD}(\mH)} \max_{t \in V(T)} \rho_\mH^*(\chi(t)).
\]
We can now define the measure of width for $\cqneg$ queries \rev{that we use}.

\begin{definition}[\rev{Fractional signed hyperorder width}]
\label{def:sfhow}
    Let $\smH=([n],\mE^+,\mE^-)$ be a signed hypergraph. The \rev{fractional signed hyperorder width} of $\smH$ is defined as
    \[
    \sfhow(\smH) := \min_{\sigma}\max_{S\subseteq \mE^-} \max_{i \in [n]} \rho_{\mH^S}^*(\chi_{\sigma}^S(\sigma_i))
    \]
\end{definition}
\rev{This is equivalent to the definition in~\cite{CapelliIrwin24,CapelliCarmeliIrwinSalvati26}, which eliminates $\sigma_i$ by removing it from every hyperedge instead of removing the hyperedges containing it. Both give the same bags $\chi_\sigma^S(\sigma_i)$.}

\begin{example}
    Consider the query $Q(x,y,z) \leftarrow R(x,y) \wedge T(y,z) \wedge \neg S(x,z)$ with variable order $\sigma=(x,y,z)$. When $S=\emptyset$, all the bags have width $1$; when $S=\{S\}$, $x$ gains the variable $z$ as a neighbor, meaning $\sfhow(Q)=3/2$.
    \begin{center}
    \scalebox{0.8}{\begin{tikzpicture}[
    scale=0.85,
    vtx/.style={circle, draw, inner sep=3pt, font=\small},
    pedge/.style={draw=black!70, thick},
    neg/.style={dashed, draw=pOrange, thick},
    bag/.style={draw, rounded corners=3pt, inner sep=4pt, font=\small,
                fill=gray!8, draw=pGray!60},
    obag/.style={draw, rounded corners=3pt, inner sep=4pt, font=\small,
                fill=pOrange!10, draw=pOrange!70},
    edge/.style={draw=gray!60, thick},
    slab/.style={font=\small\bfseries, anchor=east},
]

\node[vtx] (qx) at (0.0, -1.85) {$x$};
\node[vtx] (qy) at (1.3, -1.85) {$y$};
\node[vtx] (qz) at (0.65,-0.75) {$z$};
\draw[pedge] (qx)--(qy) node[midway, below, font=\scriptsize]{$R$};
\draw[pedge] (qy)--(qz) node[midway, right, font=\scriptsize]{$T$};
\draw[neg] (qx)--(qz) node[midway, left,  font=\scriptsize]{$S$};

\draw[gray!30, thick] (2.7,  0.12) -- (2.7, -2.75);
\draw[gray!30, thick] (7.3,  0.12) -- (7.3, -2.75);

\node[slab, text=pGray]   at (4.20, -0.22) {$S=\emptyset$};
\node[bag] (z0) at (5.00, -0.30) {$\{\underline{z}\}$};
\node[bag] (y0) at (5.00, -1.25) {$\{\underline{y},z\}$};
\node[bag] (x0) at (5.00, -2.20) {$\{\underline{x},y\}$};
\draw[edge] (z0)--(y0); \draw[edge] (y0)--(x0);

\node[slab, text=pOrange] at (9.35, -0.22) {$S=\{S\}$};
\node[obag] (z1) at (10.10, -0.30) {$\{\underline{z}\}$};
\node[obag] (y1) at (10.10, -1.25) {$\{\underline{y},z\}$};
\node[obag] (x1) at (10.10, -2.20) {$\{\underline{x},y,z\}$};
\draw[edge] (z1)--(y1); \draw[edge] (y1)--(x1);

\end{tikzpicture}}
    \end{center}
\end{example}

\subsubsection{\texorpdfstring{\rev{Comparison to Other Measures}}{Comparison to Other Measures}}\mbox{}\\
\rev{The fractional signed hyperorder width can be compared to other width measures. See~\cite{CapelliIrwin24,CapelliCarmeliIrwinSalvati26} for more on this. For completeness, we include proofs of the lemmas below in \apporfull{sec:app:missing-proofs-sec3}.}

\begin{restatable}[\rev{\cite{CapelliIrwin24,CapelliCarmeliIrwinSalvati26}}]{lemma}{sfhwequalsfhwlemma}
\label{lemma:sfhwequalsfhw}
Let $\smH=([n],\mE^+,\emptyset)$ be a signed hypergraph with no negative hyperedges, and $\mH=([n],\mE^+)$.
Then $\sfhow(\smH) = \textsf{fhw}(\mH)$.
\end{restatable}

Recall that a signed hypergraph $\smH$ is {\em signed-acyclic} if for every $S \subseteq \mE^-$, the hypergraph $\mH^S$ is $\alpha$-acyclic.

\begin{restatable}[\rev{\cite{CapelliIrwin24,CapelliCarmeliIrwinSalvati26}}]{lemma}{signedacyclicwidthonelemma}
\label{lemma:signedacyclicwidth1}
    A signed hypergraph $\smH$ is signed-acyclic if and only if $\sfhow(\smH) = 1$.
\end{restatable}

\rev{The width $\sfhow$ can also be compared to} $\beta\textsf{-fhw}$~\cite{GottlobPichler01}.\footnote{To be precise, Gottlob and Pichler~\cite{GottlobPichler01} defined the weaker notion of $\beta$-hypertree width ($\beta\textsf{-hw}$), where covers are restricted to be integral. Since $\beta\textsf{-fhw} \leq \beta\textsf{-hw}$, we will use the fractional version instead for comparison.} This is defined by first choosing the subset $S \subseteq \mE^-$, and then picking the best variable order for $\mH^S$:
  \[
    \beta\textsf{-fhw}(\smH) := \max_{S\subseteq \mE^-} \textsf{fhw}(\mH^S) = \max_{S\subseteq \mE^-} \min_{\sigma} \max_i \rho_{\mH^S}^*(\chi_{\sigma}^S(\sigma_i)).
    \]
This is a critical difference: in $\beta\textsf{-fhw}$, we can pick a different variable order for each $S$, while in $\sfhow$, we need to pick one variable order that works well for all $S$.

\begin{restatable}[\rev{\cite{CapelliIrwin24,CapelliCarmeliIrwinSalvati26}}]{lemma}{betafhwlesfhwlemma}
\label{lemma:betafhwlesfhw}
 For any signed hypergraph $\smH$, $\beta\textsf{-fhw}(\smH) \leq \sfhow(\smH)$.
\end{restatable}

\rev{The example below shows a query where the above inequality is strict.}

\begin{example} \label{example:NegTriangle}
    Consider the all-negative triangle query
    \[
    Q(x,y,z) \leftarrow R_1(x) \wedge R_2(y) \wedge R_3(z) \wedge \neg S_1(x,y) \wedge \neg S_2(y,z) \wedge \neg S_3(x,z).
    \]
    Here, $\sfhow(Q)=2$, but $\beta\textsf{-fhw}(Q)=3/2$. For any $\sigma$, the first variable is incident to two negative edges; taking those as $S$ forces the first bag to contain all three variables. The orange edges show the bad $S$ for each choice of first variable.
    \begin{center}
    \scalebox{0.8}{\begin{tikzpicture}[
    scale=0.85,
    vtx/.style={circle, draw, inner sep=2.5pt, font=\small},
    neg/.style={dashed, gray!50, thick},
    bad/.style={dashed, pOrange, very thick},
    slab/.style={font=\small\bfseries, anchor=east},
]

\def\colA{0.0}
\def\colB{5.6}
\def\colC{11.2}

\draw[gray!30, thick] (3.85, 1.10) -- (3.85, -0.80);
\draw[gray!30, thick] (9.45, 1.10) -- (9.45, -0.80);

\begin{scope}[xshift=\colA cm]
    \node[slab]                    at (0.80, 0.90) {$\sigma=(x,\ldots)$};
    \node[slab, text=pOrange] at (0.80,-0.30) {$S=\{S_1,S_3\}$};
    \node[vtx] (x) at (2.40,  0.90) {$x$};
    \node[vtx] (y) at (1.60, -0.30) {$y$};
    \node[vtx] (z) at (3.20, -0.30) {$z$};
    \draw[bad] (x)--(y) node[midway, left,  font=\scriptsize]{$S_1$};
    \draw[neg] (y)--(z) node[midway, below, font=\scriptsize]{$S_2$};
    \draw[bad] (x)--(z) node[midway, right, font=\scriptsize]{$S_3$};
\end{scope}

\begin{scope}[xshift=\colB cm]
    \node[slab]                    at (0.80, 0.90) {$\sigma=(y,\ldots)$};
    \node[slab, text=pOrange] at (0.80,-0.30) {$S=\{S_1,S_2\}$};
    \node[vtx] (x) at (2.40,  0.90) {$x$};
    \node[vtx] (y) at (1.60, -0.30) {$y$};
    \node[vtx] (z) at (3.20, -0.30) {$z$};
    \draw[bad] (x)--(y) node[midway, left,  font=\scriptsize]{$S_1$};
    \draw[bad] (y)--(z) node[midway, below, font=\scriptsize]{$S_2$};
    \draw[neg] (x)--(z) node[midway, right, font=\scriptsize]{$S_3$};
\end{scope}

\begin{scope}[xshift=\colC cm]
    \node[slab]                    at (0.80, 0.90) {$\sigma=(z,\ldots)$};
    \node[slab, text=pOrange] at (0.80,-0.30) {$S=\{S_2,S_3\}$};
    \node[vtx] (x) at (2.40,  0.90) {$x$};
    \node[vtx] (y) at (1.60, -0.30) {$y$};
    \node[vtx] (z) at (3.20, -0.30) {$z$};
    \draw[neg] (x)--(y) node[midway, left,  font=\scriptsize]{$S_1$};
    \draw[bad] (y)--(z) node[midway, below, font=\scriptsize]{$S_2$};
    \draw[bad] (x)--(z) node[midway, right, font=\scriptsize]{$S_3$};
\end{scope}

\end{tikzpicture}}
    \end{center}
\end{example}

\rev{The above example is of particular interest to us: even though $\sfhow(Q)=2$, we will see in Section~\ref{sec:multi-td} that, by using multiple variable orders, the query in Example~\ref{example:NegTriangle} can be enumerated with constant delay after $O(|I|^{3/2})$ preprocessing time (Theorem~\ref{theorem:kcycle}).}

\rev{Capelli et al.~\cite{CapelliIrwin24,CapelliCarmeliIrwinSalvati26} also show that $\sfhow \leq \textsf{nsw}$ when all positive atoms are unary, where $\textsf{nsw}$ is the nest-set width of Lanzinger~\cite{NestSetWidth}. This inequality can be strict: for negative $k$-cycle queries, $\textsf{nsw} \geq k-1$~\cite{NestSetWidth}, while $\sfhow = 2$.}

\section{Computing Conjunctive Queries with Negation}

In this section, we will prove the main result of the paper, by presenting an enumeration algorithm for Conjunctive Queries with Negation.%
\begin{restatable}{theorem}{mainthm}
\label{thm:main}
    Given a full $\cqneg$ $Q$ with a signed hypergraph $\smH$ and input $I$, there is an algorithm that enumerates $Q(I)$ with constant delay after a preprocessing phase with runtime $O(|I|^{\sfhow(\smH)})$.
\end{restatable}

The algorithm works in a variable-at-a-time fashion: given a variable order $\sigma$, we create a sequence of \textit{data structures} -- one per variable -- that enumerate values over the next variable in the order, given a tuple over the previous variables. The data structure is similar to the one used in previous work for signed-acyclic queries~\cite{SignedAcyclic}; it consists of a hash table mapping to linked lists of values of the next variable. Additionally, there are \textit{skip links} over the linked lists, which enable skipping values which are prohibited because of negative atoms.

\subsection{Enumerating via Reduction to a Smaller Query}

To facilitate the variable-at-a-time enumeration, we define a reduction which, given a query over $\smH$, creates a reduced query with one fewer variable. This way, given a variable order $\sigma$, we obtain a sequence of signed hypergraphs. This is analogous to the sequence of hypergraphs in Section~\ref{sec:varorderrefresh}, but for signed hypergraphs.

Consider a signed hypergraph $\smH$ and let $\sigma$ be a variable order. We define $\smH_1=\smH$, while $\smH_{i+1}=(V_{i+1},\mE_{i+1}^+,\mE_{i+1}^-)$ is obtained inductively from $\smH_{i}=(V_{i}, \mE_{i}^+, \mE_{i}^-)$ as follows:
\begin{enumerate}
    \item remove the vertex $\sigma_{i}$;
    \item remove any hyperedge $K$ such that $\sigma_{i} \in K$;
    \item add a positive hyperedge $U_i^\emptyset \setminus\{\sigma_{i}\}$, where $U_i^\emptyset := \bigcup_{K \in \partial_{\mH_{i}^\emptyset} \sigma_{i}} K$;
    \item for each non-empty $S\subseteq \{K \in \mE_{i}^- : \sigma_i \in K\}$: add a negative hyperedge $U_i^S \setminus\{\sigma_{i}\}$, where $U_i^S := \bigcup_{K \in \partial_{\mH_{i}^S} \sigma_{i}} K$.
\end{enumerate}

Note that $\smH_{n+1}$ is well-defined and has no vertices.
Further, we maintain on each $K \in \mE_i^-$ an \textit{annotation} $\mA_K \subseteq \mE^-$. Annotations are defined recursively: an input $K\in \mE^-$ has annotation $\mA_K=\{K\}$, while $U_i^S$ has annotation $\mA_{U_i^S}=\cup_{K\in S}\mA_K$. Using our notation so far, $\mH_i^S$ is the hypergraph with vertices $V_i$, and edges $\mE_i^+ \cup S$, where $S \subseteq \mE_i^-$.

\begin{example}
    Let $\smH$ be a signed hypergraph with variables $V=\{x,y,w,z\}$, a unary positive atom for each variable, and negative atoms $\mE^-=\{{\color{pRed}R(x,y,w)},\, {\color{pBlue}S(x,z,w)},\, {\color{pOrange}T(x,z)}\}$. Consider a variable order $\sigma$ with $\sigma_1 = x$. The figure below visualizes one elimination step, from $\smH_1$ to $\smH_2$, alongside the hyperedges and annotations of each negative atom in the two queries:
    \begin{center}
    \small
    \resizebox{\textwidth}{!}{\begin{tabular}{@{}c@{\hspace{1em}}|@{\hspace{1em}}c@{}}
%
\begin{tikzpicture}[
    baseline=(current bounding box.center),
    >=stealth,
    var/.style ={circle, draw=black, fill=white, line width=0.5pt,
                 inner sep=1pt, minimum size=14pt, font=\scriptsize},
    elim/.style={circle, draw=gray!50, fill=gray!12, line width=0.4pt,
                 inner sep=1pt, minimum size=14pt, font=\scriptsize\itshape, text=gray!55},
    title/.style={font=\scriptsize\bfseries, anchor=south},
    ringR/.style ={draw=pRed,    line width=0.7pt, rounded corners=6pt},
    ringS/.style ={draw=pBlue,   line width=0.7pt, rounded corners=6pt},
    ringT/.style ={draw=pOrange,  line width=0.7pt, rounded corners=5pt},
    ringRS/.style={draw=pPurple, line width=0.7pt, rounded corners=8pt},
]

\pgfmathsetmacro{\sx}{0.70}     
\pgfmathsetmacro{\sy}{0.70}     
\pgfmathsetmacro{\rNode}{0.18}  

\pgfmathsetmacro{\ra}{0.12}
\pgfmathsetmacro{\rb}{0.24}
\pgfmathsetmacro{\rc}{0.36}
\pgfmathsetmacro{\rd}{0.44}
\pgfmathsetmacro{\na}{0.16}
\pgfmathsetmacro{\nb}{0.38}

\pgfmathsetmacro{\panelSep}{4.2}

\begin{scope}[xshift=0cm]
  \node[title] at (0, 1.42) {$\smH_1=\smH$};

  \coordinate (TL) at (-\sx,  \sy);
  \coordinate (TR) at ( \sx,  \sy);
  \coordinate (BL) at (-\sx, -\sy);
  \coordinate (BR) at ( \sx, -\sy);

  \pgfmathsetmacro{\OX}{\sx + \rNode + \rc}
  \pgfmathsetmacro{\OY}{\sy + \rNode + \rc}
  \pgfmathsetmacro{\NXR}{\sx - \rNode - \nb}
  \pgfmathsetmacro{\NYR}{\sy - \rNode - \nb}
  \draw[ringR]
        (-\OX,  \OY) --
        ( \NXR, \OY) --
        ( \NXR, \NYR) --
        ( \OX,  \NYR) --
        ( \OX, -\OY) --
        (-\OX, -\OY) -- cycle;

  \pgfmathsetmacro{\oX}{\sx + \rNode + \rb}
  \pgfmathsetmacro{\oY}{\sy + \rNode + \rb}
  \pgfmathsetmacro{\NXS}{-\sx + \rNode + \na}
  \pgfmathsetmacro{\NYS}{-\sy + \rNode + \na}
  \draw[ringS]
        (-\oX,  \oY) --
        ( \oX,  \oY) --
        ( \oX, -\oY) --
        ( \NXS,-\oY) --
        ( \NXS,\NYS) --
        (-\oX, \NYS) -- cycle;

  \pgfmathsetmacro{\tL}{-\sx - \rNode - \ra}
  \pgfmathsetmacro{\tR}{ \sx + \rNode + \ra}
  \pgfmathsetmacro{\tU}{ \sy + \rNode + \ra}
  \pgfmathsetmacro{\tD}{ \sy - \rNode - \ra}
  \draw[ringT] (\tL,\tD) rectangle (\tR,\tU);

  \node[var] at (TL) {$x$};
  \node[var] at (TR) {$z$};
  \node[var] at (BL) {$y$};
  \node[var] at (BR) {$w$};
\end{scope}

\pgfmathsetmacro{\arrL}{\sx + \rNode + \rc + 0.20}
\pgfmathsetmacro{\arrR}{\panelSep - \sx - \rNode - \rd - 0.20}
\draw[->, line width=0.7pt, gray!70]
      (\arrL, 0.0) -- node[above, font=\tiny, text=black]{eliminate $x$} (\arrR, 0.0);

\begin{scope}[xshift=\panelSep cm]
  \node[title] at (0, 1.42) {$\smH_2$};

  \coordinate (TL) at (-\sx,  \sy);
  \coordinate (TR) at ( \sx,  \sy);
  \coordinate (BL) at (-\sx, -\sy);
  \coordinate (BR) at ( \sx, -\sy);

  \pgfmathsetmacro{\OX}{\sx + \rNode + \rd}
  \pgfmathsetmacro{\OY}{\sy + \rNode + \rd}
  \pgfmathsetmacro{\NX}{-\sx + \rNode + \nb}
  \pgfmathsetmacro{\NY}{ \sy - \rNode - \nb}
  \draw[ringRS]
        (-\OX, -\OY) --
        ( \OX, -\OY) --
        ( \OX,  \OY) --
        ( \NX,  \OY) --
        ( \NX,  \NY) --
        (-\OX,  \NY) -- cycle;

  \pgfmathsetmacro{\oX}{\sx + \rNode + \ra}
  \pgfmathsetmacro{\bU}{-\sy + \rNode + \ra}
  \pgfmathsetmacro{\bD}{-\sy - \rNode - \ra}
  \draw[ringR] (-\oX, \bD) rectangle ( \oX, \bU);

  \pgfmathsetmacro{\sL}{ \sx - \rNode - \rb}
  \pgfmathsetmacro{\sR}{ \sx + \rNode + \rb}
  \pgfmathsetmacro{\sUu}{ \sy + \rNode + \rb}
  \pgfmathsetmacro{\sDd}{-\sy - \rNode - \rb}
  \draw[ringS] (\sL, \sDd) rectangle (\sR, \sUu);

  \pgfmathsetmacro{\Rut}{\rNode + \ra + 0.06}
  \draw[ringT] (TR) circle [radius=\Rut];

  \node[elim] at (TL) {$x$};
  \node[var]  at (TR) {$z$};
  \node[var]  at (BL) {$y$};
  \node[var]  at (BR) {$w$};
\end{scope}

\end{tikzpicture}
    &
{\renewcommand{\arraystretch}{0.95}%
\begin{tabular}[c]{@{}ll@{\hspace{0.8em}}!{\color{black!12}\vrule width 0.3pt}@{\hspace{0.8em}}ll@{\hspace{0.8em}}!{\color{black!12}\vrule width 0.3pt}@{\hspace{0.8em}}ll@{}}
\multicolumn{2}{c}{$\smH_1=\smH$} & \multicolumn{4}{c}{$\smH_2$} \\
\cmidrule(r){1-2}\cmidrule(l){3-6}
\multicolumn{1}{@{}l}{Hyperedges} & \multicolumn{1}{l@{\hspace{0.8em}}}{$\mA$} & \multicolumn{1}{@{\hspace{0.8em}}l}{Hyperedges} & \multicolumn{1}{l@{\hspace{0.8em}}}{$\mA$} & \multicolumn{1}{@{\hspace{0.8em}}l}{Hyperedges} & $\mA$ \\
\midrule
{\color{pRed}    $(x,y,w)$} & {\color{pRed}    $\{R\}$} & {\color{pRed}    $(y,w)$}   & {\color{pRed}    $\{R\}$}   & {\color{pPurple} $(y,z,w)$} & {\color{pPurple} $\{R,T\}$}   \\
{\color{pBlue}   $(x,z,w)$} & {\color{pBlue}   $\{S\}$} & {\color{pBlue}   $(z,w)$}   & {\color{pBlue}   $\{S\}$}   & {\color{pBlue}   $(z,w)$}   & {\color{pBlue}   $\{S,T\}$}   \\
{\color{pOrange}  $(x,z)$}   & {\color{pOrange}  $\{T\}$} & {\color{pOrange}  $(z)$}     & {\color{pOrange}  $\{T\}$}   & {\color{pPurple} $(y,z,w)$} & {\color{pPurple} $\{R,S,T\}$} \\
                                    &                                   & {\color{pPurple} $(y,z,w)$} & {\color{pPurple} $\{R,S\}$} &                                     &                                       \\
\end{tabular}}

    \end{tabular}}
    \end{center}
    In general, there can be multiple atoms with the same hyperedge variables and same annotations.
\end{example}

We are now ready to present the main algorithm (Algorithm~\ref{algorithm:MainAlgorithm}), which works in two phases: a preprocessing phase and an enumeration phase. In the preprocessing phase, we use a variable order $\sigma$ witnessing $\sfhow(\smH)$, starting from $\sigma_1$ to $\sigma_n$, to create a sequence of enumeration data structures $\mathbb{D}_1, \ldots, \mathbb{D}_n$. In step $i$, we call $\textproc{BuildEnumDS}$, which builds $\mathbb{D}_i$ given $\smH_i$ and $\sigma_i$. It also creates the next reduced query $\smH_{i+1}$ with a corresponding input. In the enumeration phase, we use the enumeration data structures, following the reverse variable order, to enumerate the output tuples.

\begin{algorithm}[H]
\caption{Main Algorithm}
\label{algorithm:MainAlgorithm}
\footnotesize
\begin{algorithmic}[1]
\Procedure{SolveQuery}{$\smH, \sigma$}
\For{$i=1, \dots, n$} \Comment{Preprocessing phase}
    \State $(\mathbb{D}_i, \smH_{i+1}) \gets \textproc{BuildEnumDS}(\smH_i, \sigma_i)$
\EndFor
\For{$t_n \in \textproc{Enumerate}(\mathbb{D}_n, ())$} \Comment{Enumeration phase}
\For{$t_{n-1} \in \textproc{Enumerate}(\mathbb{D}_{n-1}, t_n)$}
\State$\dots$
\For{$t_1 \in \textproc{Enumerate}(\mathbb{D}_1, t_{2} \cdot t_3 \cdot \dots \cdot t_n)$}
\State $\textproc{Emit}(t_{1} \cdot\dots\cdot t_n)$
\EndFor
\EndFor
\EndFor
\EndProcedure
\end{algorithmic}
\end{algorithm}

\paragraph{$\cqneg$ with Projections} Our algorithm can be extended to $\cqneg$ queries that are not full by considering only variable orders $\sigma$ consistent with the free variables $F$: this means that $F$ forms a suffix of $\sigma$. In this case, we modify Algorithm~\ref{algorithm:MainAlgorithm} to stop the enumeration once all variables in $F$ have been enumerated. Since the variable order is consistent with $F$, we will never need to enumerate a variable that must be projected out. The formula for the width measure is the same as $\sfhow$, with the only difference that $\min_\sigma$ does not go over all variable orders, but only the ones consistent with $F$.

\subsection{A Data Structure for Constant-Delay Enumeration}

We now describe the enumeration data structure $\mathbb{D}_i$ for enumerating a variable $\sigma_{i} = x$ after we have enumerated the suffix $\Sigma_i = \{\sigma_{i+1}, \ldots, \sigma_n\}$. Recall that at the point we want to enumerate $x$, we have fixed values for the suffix variables.

Given a tuple $t_X$ over variables $X \subseteq [n]$ and a tuple $t_Y$ over $Y \subseteq [n]$, where $X \cap Y = \emptyset$, we will sometimes denote by $t_X \cdot t_Y$ the concatenated tuple over variables $X \cup Y$. Given two tuples $t_X$ over variables $X$ and $t_Y$ over $Y$, we say that $t_X$ and $t_Y$ are \textit{consistent}, written $t_X \sim t_Y$, if for every variable $v \in X \cap Y$, $t_X.v = t_Y.v$, i.e., they assign the same value to every variable they have in common.

\paragraph{Hash tables of linked lists}
Let $R_K$ be a relation over variables $K$, with $\{x\} \subseteq K \subseteq \Sigma_i \cup \{x\}$ (we will see in Section~\ref{subsec:buildenumds} how $R_K$ is computed). We will store $R_K$ as a hash table: the hash keys are the tuples in $\pi_{K\setminus\{x\}}R_K$, and a bucket $b\in \pi_{K\setminus\{x\}}R_K$ has a linked list $\mL$ of the values over $x$ such that $b \cdot x$ is a tuple in $R_K$. The linked list starts at the bucket key $b$ and ends with a sentinel node. For a node $a \in \mL$, we write $\textsf{prev}_{\mathbb{D}}(a)$ and $\textsf{next}_{\mathbb{D}}(a)$ for its predecessor and successor. The predecessor of the first value is $b$, and the successor of the last value is the sentinel. For a non-empty set $A$ of nodes in $\mL$, we write $\min_{\mathbb{D}} A$ and $\max_{\mathbb{D}} A$ for the first and last node of $A$ in the linked-list order.

\paragraph{Skip links} Let $\mL$ be one of the linked lists in the hash table (say bucket $b$), which we view as a set of linked list nodes. A \textit{skip link} $s$ is a quadruple $(s_{start}, s_{end}, s_{label}, \textsf{prov}(s))$, where $s_{start}, s_{end} \in \mL$ ($s_{start}$ may be the bucket key and $s_{end}$ may be the sentinel), $s_{label}$ is a tuple over some set of variables $Y \subseteq \Sigma_i \setminus K$, and $\textsf{prov}(s)$ is the \textit{provenance}: the set of negative atoms that were used in creating $s$ (we omit $\textsf{prov}(s)$ when not relevant). Semantically, a skip link from $s_{start}$ to $s_{end}$ with label $s_{label}$ indicates that given a tuple $t$ over the suffix $\Sigma_i$ consistent with $s_{label}$ and the bucket $b$, the extensions of $t$ to $x$-values \textit{strictly} between $s_{start}$ and $s_{end}$ in the linked list are not part of any output tuple, and should not be enumerated (thus can be skipped), because each such extension is \textit{blocked} from being an output by some negative atom in $\textsf{prov}(s)$. We say $s$ \textit{covers} a set $A \subseteq \mL$ of consecutive nodes in the linked list if $s_{start}$ is strictly before every value in $A$ and $s_{end}$ is strictly after every value in $A$.

\begin{example}
\label{example:enumDS}
    Consider the query $Q(x,y,z,w) \gets R(x,y) \wedge \neg T(x,y,z) \wedge \neg W(x,w) \wedge {\color{gray} U(z) \wedge V(w)}$ ($U$ and $V$ are atoms needed for safety but are not of interest in the example). The following enumeration data structure for $x$ with suffix $(y,z,w)$ will be constructed. The color of the skip links indicates the provenance.
    \begin{center}
    \resizebox{0.8\textwidth}{!}{
\begin{tikzpicture}[
    thd/.style={draw=gray!40, fill=gray!15, minimum width=0.75cm,
                minimum height=0.40cm, font=\scriptsize\bfseries, inner sep=0, outer sep=0},
    tbl/.style={draw=gray!35, fill=white, minimum width=0.75cm,
                minimum height=0.40cm, font=\scriptsize, inner sep=0, outer sep=0},
    cell/.style={draw=gray!65!black, fill=white, minimum width=1.0cm,
                 minimum height=0.48cm, font=\small, inner sep=2pt},
    endcell/.style={draw=gray!65!black, fill=gray!12, minimum width=0.6cm,
                 minimum height=0.48cm, font=\scriptsize\sffamily, text=gray!55!black, inner sep=1pt},
    ptr/.style={->, >=stealth, line width=0.5pt, gray!60!black},
    hptr/.style={->, >=stealth, thin, dashed, gray!50},
    sk/.style={->, >=stealth, semithick, rounded corners=4pt},
]
\colorlet{cw}{pBlue}
\colorlet{cz}{pOrange}
\colorlet{czw}{pPurple}
\colorlet{grp}{pGray}

\node[font=\scriptsize\itshape] at (0.75, 0.55){$R(x,y)$};
\node[thd] at (0.375, 0.20){$x$};    \node[thd] at (1.125, 0.20){$y$};
\node[tbl] at (0.375,-0.20){$a_1$};  \node[tbl] at (1.125,-0.20){$b_1$};
\node[tbl] at (0.375,-0.60){$a_2$};  \node[tbl] at (1.125,-0.60){$b_1$};
\node[tbl] at (0.375,-1.00){$a_3$};  \node[tbl] at (1.125,-1.00){$b_1$};
\node[tbl] at (0.375,-1.50){$a_1$};  \node[tbl] at (1.125,-1.50){$b_2$};
\node[tbl] at (0.375,-1.90){$a_2$};  \node[tbl] at (1.125,-1.90){$b_2$};
\node[tbl] at (0.375,-2.30){$a_4$};  \node[tbl] at (1.125,-2.30){$b_2$};
\node[tbl] at (0.375,-2.70){$a_5$};  \node[tbl] at (1.125,-2.70){$b_2$};
\node[tbl] at (0.375,-3.10){$a_6$};  \node[tbl] at (1.125,-3.10){$b_2$};
\draw[grp, dashed, rounded corners=2pt] (-0.06, 0.02) rectangle (1.56,-1.21);
\draw[grp, dashed, rounded corners=2pt] (-0.06,-1.29) rectangle (1.56,-3.32);

\node[font=\scriptsize\itshape] at (3.125, 0.55){$\neg T(x,y,z)$};
\node[thd] at (2.375, 0.20){$x$};    \node[thd] at (3.125, 0.20){$y$};    \node[thd] at (3.875, 0.20){$z$};
\node[tbl] at (2.375,-0.20){$a_4$};  \node[tbl] at (3.125,-0.20){$b_2$};  \node[tbl] at (3.875,-0.20){$c_1$};
\node[tbl] at (2.375,-0.60){$a_5$};  \node[tbl] at (3.125,-0.60){$b_2$};  \node[tbl] at (3.875,-0.60){$c_1$};
\draw[cz, dashed, rounded corners=2pt] (2.00,-0.02) rectangle (4.25,-0.82);

\node[font=\scriptsize\itshape] at (2.75,-1.12){$\neg W(x,w)$};
\node[thd] at (2.375,-1.50){$x$};    \node[thd] at (3.125,-1.50){$w$};
\node[tbl] at (2.375,-1.90){$a_2$};  \node[tbl] at (3.125,-1.90){$d_1$};
\draw[cw, dashed, rounded corners=2pt] (2.00,-1.72) rectangle (3.50,-2.08);

\draw[gray!40, fill=gray!10] (5.00, 0.22) rectangle (6.50,-0.22);
\draw[grp, dashed, rounded corners=2pt] (5.00, 0.22) rectangle (6.50,-0.22);
\node[font=\scriptsize] at (5.75, 0.00){$b_1$};

\node[font=\normalsize,gray!40] at (5.75,-1.10){$\vdots$};

\draw[gray!40, fill=gray!10] (5.00,-1.98) rectangle (6.50,-2.42);
\draw[grp, dashed, rounded corners=2pt] (5.00,-1.98) rectangle (6.50,-2.42);
\node[font=\scriptsize] at (5.75,-2.20){$b_2$};

\node[cell](a1) at ( 7.75, 0.00){$a_1$};
\node[cell](a2) at ( 9.25, 0.00){$a_2$};
\node[cell](a3) at (10.75, 0.00){$a_3$};
\draw[hptr](6.50, 0.00)--(a1.west);
\draw[ptr](a1)--(a2);  \draw[ptr](a2)--(a3);
\node[endcell](ae) at (12.05, 0.00){end}; \draw[ptr](a3)--(ae);

\draw[sk,cw] (a1.north)--(7.75,0.65)--(10.75,0.65)--(a3.north);
\node[font=\scriptsize,cw,anchor=south,inner sep=1.5pt] at (9.25,0.65){$d_1$};

\node[cell](b1) at ( 7.75,-2.20){$a_1$};
\node[cell](b2) at ( 9.25,-2.20){$a_2$};
\node[cell](b3) at (10.75,-2.20){$a_4$};
\node[cell](b4) at (12.25,-2.20){$a_5$};
\node[cell](b5) at (13.75,-2.20){$a_6$};
\draw[hptr](6.50,-2.20)--(b1.west);
\draw[ptr](b1)--(b2);  \draw[ptr](b2)--(b3);  \draw[ptr](b3)--(b4);  \draw[ptr](b4)--(b5);
\node[endcell](be) at (15.05,-2.20){end}; \draw[ptr](b5)--(be);

\draw[sk,cw]  ([xshift=0.15cm]b1.north)--(7.90,-1.55)--(10.75,-1.55)--(b3.north);
\node[font=\scriptsize,cw,anchor=south,inner sep=1.5pt]  at (8.575,-1.55){$d_1$};
\draw[sk,cz]  (b2.north)--(9.25,-1.20)--(13.60,-1.20)--([xshift=-0.15cm]b5.north);
\node[font=\scriptsize,cz,anchor=south,inner sep=1.5pt]  at (11.50,-1.20){$c_1$};
\draw[sk,czw] ([xshift=-0.15cm]b1.north)--(7.60,-0.85)--(13.90,-0.85)--([xshift=0.15cm]b5.north);
\node[font=\scriptsize,czw,anchor=south,inner sep=1.5pt] at (10.75,-0.85){$c_1d_1$};

\draw[sk,cw]  (5.00,-3.20)--(5.70,-3.20); \node[anchor=west,font=\scriptsize,cw]  at (5.78,-3.20){Skip $\{W\}$};
\draw[sk,cz]  (7.60,-3.20)--(8.30,-3.20); \node[anchor=west,font=\scriptsize,cz]  at (8.38,-3.20){Skip $\{T\}$};
\draw[sk,czw] (10.20,-3.20)--(10.90,-3.20); \node[anchor=west,font=\scriptsize,czw] at (10.98,-3.20){Skip $\{T,W\}$};

\end{tikzpicture}}
    \end{center}
\end{example}

For each node in a linked list, we have local hash tables storing the skip links that start/end at that node, where the key is the skip link label. In this way, we can find skip links consistent with the current suffix in constant time, by checking for relevant subsets of the suffix as keys in the hash table.
Given a suffix tuple $(a_{i+1}, \ldots, a_n)$, we use the data structure as follows. First, we use the hash table to find the (unique) bucket $b$ consistent with the suffix tuple, giving a linked list $\mL$. Then, at every step of enumeration, we check all skip links starting from the current node $a$ in $\mL$ consistent with the suffix tuple, and follow the skip link with the longest skip. If no such skip link exists, we go to $\textsf{next}_\mathbb{D}(a)$. In Example~\ref{example:enumDS}, the \textproc{Enumerate} iterator on suffix $(b_2,c_2,d_1)$ enumerates $a_1\rightarrow a_4 \rightarrow a_5 \rightarrow a_6$, following the $d_1$ skip link. On suffix $(b_2,c_1,d_1)$, it enumerates $a_1\rightarrow a_6$, following the $c_1 d_1$ skip link.

We will call two skip links $s$ and $s'$ \textit{non-containingly overlapping} if $s_{start} < s'_{start} < s_{end} < s'_{end}$ or $s'_{start} < s_{start} < s'_{end} < s_{end}$.
Pairs of skip links $s$ and $s'$ that \textit{both} have consistent labels and are non-containingly overlapping will be of particular importance, as we want to skip both of them at once for correctness and constant delay (if we follow the first, we end up at a node which should also be skipped). Therefore, we add \textit{merged} skip links, which span across all the relevant skip links, enabling reaching the next value to output by following at most one skip link. In Example~\ref{example:enumDS}, the $c_1 d_1$ skip link in bucket $b_2$ is a merged skip link, which merges the $c_1$ and $d_1$ skip links, which are non-containingly overlapping and have consistent labels. The process of merging skip links is a critical technical component of the data structure construction.

The construction of the data structure for a variable $x$ will be guided by the bags $\chi_\sigma^S(x)$ for different subsets $S$ of negative edges: the hash table is built over the query over positive atoms induced on $\chi_\sigma^\emptyset(x)$. For each $S$, we will add skip links with provenance $S$, with label and bucket key described by $\chi_\sigma^S(x)\setminus\{x\}$, using a query induced on that bag.

\subsection{Building an Enumeration Data Structure}
\label{subsec:buildenumds}

This subsection describes \textproc{BuildEnumDS}, or Algorithm~\ref{algorithm:BuildEnumDS}, which builds an enumeration data structure $\mathbb{D}_i$ for variable $\sigma_i$, using the query over $\smH_i$. The procedure also creates the next reduced query $\smH_{i+1}$, including an input instance, which is used in the next step of the preprocessing phase.

\begin{algorithm}
\caption{Build Enumeration Data Structure}
\label{algorithm:BuildEnumDS}
\footnotesize
\begin{algorithmic}[1]
\Procedure{BuildEnumDS}{$\smH_i=(V_i,\mE_i^+,\mE_i^-),\sigma_i$}
\State $L^\emptyset \gets \textproc{Join}(\mH_i^\emptyset[\chi_\sigma^\emptyset(\sigma_i)])$
\State $\mathbb{D} \gets \textproc{HashTable}(L^\emptyset)$
\State // Create initial points to skip
\State $\mC \gets \{K \in \mE_i^- : \sigma_i \in K\}$ \Comment{Negative atoms containing $\sigma_i$}
\For{$R \in \mC$}
    \State $P^R \gets \textproc{Join}(\mH_i^{\{R\}}[\chi_\sigma^{\mA_R}(\sigma_i)])$ \Comment{Tuples to skip}
\EndFor
\State // Merge into skip links
\For{$S \subseteq \mC$, $S \neq \emptyset$} \Comment{Follow order of increasing $|S|$}
    \State $\mA_S \gets \bigcup_{R\in S}\mA_R$
    \State $L^S \gets \textproc{Join}(\mH_i^S[\chi_\sigma^{\mA_S}(\sigma_i)\setminus\{\sigma_i\}])$ \Comment{Possible label/bucket}
    \State $T^S \gets \emptyset$ \Comment{Skip links before merge}
    \For{$\ell\cdot B \in L^S$}
        \State $R^* \gets \arg\min_{R\in S}\deg_{P^R}(\sigma_i\mid \ell \cdot B)$
        \State $V^S \gets (\ell \cdot B) \bowtie P^{R^*}$ \Comment{Points to skip}
        \State $T^S \gets T^S \cup \bigcup_{v \in V^S} \textproc{Extend}(v, S)$
    \EndFor
    \State $\bar{T}^S \gets \textproc{Merge}(T^S)$
    \State $\mathbb{D} \gets \mathbb{D}$ with new skip links $\bar{T}^S$
\EndFor
\State Construct $\smH_{i+1}$ from $\smH_i$:
\Statex \quad $\bullet$\ remove $\sigma_i$ and every atom containing $\sigma_i$;
\Statex \quad $\bullet$\ add the positive atom $\;U_i^\emptyset \setminus \{\sigma_i\} \;:=\; \pi_{\setminus\sigma_i}\,L^\emptyset$;
\Statex \quad $\bullet$\ for each non-empty $S \subseteq \mC$, add the negative atom
\Statex \qquad\quad $U_i^S \setminus \{\sigma_i\} \;:=\; \{\, s_{label}\cdot b \,:\, s \in \bar{T}^S \text{ in bucket } b,\ s \text{ covers the entire linked list}\,\}$.
\State\Return $(\mathbb{D}, \smH_{i+1})$
\EndProcedure
\Statex
\Procedure{Extend}{$v$, $S$} \Comment{Extend $v$ into a skip link}
    \State $s_{in} \gets \arg\min\{ s' \in \mathbb{D} \mid \textsf{prov}(s') \subsetneq S,\, s'_{end} = \pi_{\sigma_i}v,\, s'_{label} \sim v \}$ \Comment{Longest incoming skip link}
    \State $s_{start} \gets (s_{in})_{start}$ if $s_{in} \neq \textsf{null}$, else $\textsf{prev}_\mathbb{D}(\pi_{\sigma_i}v)$
    \State $s_{out} \gets \arg\max\{ s' \in \mathbb{D} \mid \textsf{prov}(s') \subsetneq S,\, s'_{start} = \pi_{\sigma_i}v,\, s'_{label} \sim v \}$ \Comment{Longest outgoing skip link}
    \State $s_{end} \gets (s_{out})_{end}$ if $s_{out} \neq \textsf{null}$, else $\textsf{next}_\mathbb{D}(\pi_{\sigma_i}v)$
    \State $s_{label} \gets \pi_{\chi_\sigma^{\mA_S}(\sigma_i)\setminus\chi_\sigma^{\emptyset}(\sigma_i)}v$
    \State $\textsf{prov}(s) \gets S$
    \State\Return $\{(s_{start}, s_{end}, s_{label}, \textsf{prov}(s))\}$
\EndProcedure
\Statex
\Procedure{Merge}{$T$} \Comment{Merge non-containingly overlapping skip links}
    \For{$s, s' \in T$ non-containingly overlapping with $s_{label} = s'_{label}$}
        \State $T \gets (T \setminus \{s, s'\}) \cup \{(\min(s_{start}, s'_{start}),\, \max(s_{end}, s'_{end}),\, s_{label},\, \textsf{prov}(s))\}$
    \EndFor
    \State \Return $T$
\EndProcedure
\end{algorithmic}
\end{algorithm}

\paragraph{Intuition} We first give some intuition for the algorithm. As a first step, Lines 2--3 construct the initial hash table with linked lists.

Lines 5--7 perform the first step towards obtaining skip links. Here, we compute a set of tuples $P^R$ for each negative atom $R$ where $\sigma_i\in R$. $P^R$ describes which tuples in the hash table are skipped by the negative atom $R$. We do this through the join over $\mH_i^{\{R\}}[\chi_\sigma^{\mA_R}(\sigma_i)]$, because this join gives us tuples that exist over the positive query, but also match the negative atom $R$ -- this allows us to find tuples that would be output if not for the negative atom.

Lines 9--18 turn the skipped tuples into skip links. Note that we merge consecutive skipped tuples both when they are skipped by the same $R$ and when they are skipped by different negative atoms. Every subset $S$ of negative atoms creates a set of skip links with provenance $S$, which we sometimes call \textit{$S$-skip links}: these are the skip links that cover a region of a linked list, where each node is skipped by some negative atom in $S$. We build the $S$-skip links in order of increasing $|S|$, and will utilize the already built $S'$-skip links when constructing $S$-skip links, where $S' \subsetneq S$.

For each $S$, we compute a set $L^S$. This contains possible pairs of skip link label and hash bucket for $S$-skip links. For each such label/hash bucket pair, we then compute which values over $\sigma_i$ are skipped. To do this, we essentially perform one step of a worst-case optimal join: we pick the cheapest $P^R$ to extend to the variable $\sigma_i$ for each tuple in $L^S$, giving a set of points in the skip link. We then \textit{extend} these points into $S$-skip links, by using, for each point, incoming/outgoing $S'$-skip links, where $S'\subsetneq S$. Finally, we merge resulting skip links which are non-containingly overlapping and have the same label, and tag the resulting skip links with provenance $S$.

When we construct the next reduced query $\smH_{i+1}$, we put all hash table keys into $U_i^\emptyset\setminus\{\sigma_i\}$. This enforces that the next reduced query only creates output tuples which have positive extensions to $\sigma_i$, essentially acting like the semijoins in Yannakakis algorithm. For $U_i^S\setminus\{\sigma_i\}$ where $S \neq \emptyset$, we put the skip link labels/bucket pairs with provenance $S$ which span an entire linked list. Any tuple consistent with such a pair will not create an output over $\smH_i$, and therefore does not need to be output from $\smH_{i+1}$.


\begin{example}
    We continue Example~\ref{example:enumDS} with $Q(x,y,z,w) \gets R(x,y) \wedge \neg T(x,y,z) \wedge \neg W(x,w)$, showing how \textproc{BuildEnumDS} constructs the enumeration data structure for $x$. We only show the linked list for one hash bucket $b_2$. The relevant bags $\chi_\sigma^S(x)$, which describe the variables occurring in the hash table and in each type of skip link, are:
    \begin{center}
    \begin{tikzpicture}[
    scale=0.75,
    bag/.style={draw, rounded corners=3pt, inner sep=3pt, font=\scriptsize,
                fill=gray!8, draw=pGray},
    bagS/.style={draw, rounded corners=3pt, inner sep=3pt, font=\scriptsize,
                fill=pBlue!10, draw=pBlue},
    bagT/.style={draw, rounded corners=3pt, inner sep=3pt, font=\scriptsize,
                fill=pOrange!10, draw=pOrange},
    bagST/.style={draw, rounded corners=3pt, inner sep=3pt, font=\scriptsize,
                fill=pPurple!10, draw=pPurple},
    slab/.style={font=\scriptsize\bfseries, anchor=east},
]
\colorlet{cw}{pBlue}
\colorlet{cz}{pOrange}
\colorlet{czw}{pPurple}

\node[slab, text=pGray] at (1.00, 0) {$S=\emptyset$};
\node[bag] (b0) at (1.75, 0) {$\{\underline{x},y\}$};

\node[slab, text=cw] at (4.80, 0) {$S=\{W\}$};
\node[bagS] (b1) at (5.65, 0) {$\{\underline{x},y,w\}$};

\node[slab, text=cz] at (8.70, 0) {$S=\{T\}$};
\node[bagT] (b2) at (9.50, 0) {$\{\underline{x},y,z\}$};

\node[slab, text=czw] at (12.80, 0) {$S=\{T,W\}$};
\node[bagST] (b3) at (13.80, 0) {$\{\underline{x},y,z,w\}$};
\end{tikzpicture}
    \end{center}
    Lines~5--7: compute $P^W$ and $P^T$:
    \begin{center}
    \scalebox{0.8}{
\begin{tikzpicture}[
    cell/.style={draw=gray!65!black, fill=white, minimum width=0.9cm,
                 minimum height=0.44cm, font=\small, inner sep=2pt},
    keycell/.style={draw=pGray, dashed, fill=gray!10, minimum width=0.8cm,
                 minimum height=0.44cm, font=\small, inner sep=2pt},
    endcell/.style={draw=gray!65!black, fill=gray!12, minimum width=0.55cm,
                 minimum height=0.44cm, font=\scriptsize\sffamily, text=gray!55!black, inner sep=1pt},
    hptr/.style={->, >=stealth, thin, dashed, gray!50},
    ptr/.style={->, >=stealth, line width=0.5pt, gray!60!black},
    pt/.style={circle, fill, inner sep=1.5pt},
    lbl/.style={font=\scriptsize},
]
\colorlet{cs}{pBlue}
\colorlet{cz}{pOrange}

\def\dx{1.4}
\def\xoff{1.6}

\node[cell](c1) at (\xoff+0*\dx, 0){$a_1$};
\node[cell](c2) at (\xoff+1*\dx, 0){$a_2$};
\node[cell](c3) at (\xoff+2*\dx, 0){$a_4$};
\node[cell](c4) at (\xoff+3*\dx, 0){$a_5$};
\node[cell](c5) at (\xoff+4*\dx, 0){$a_6$};
\draw[ptr](c1)--(c2); \draw[ptr](c2)--(c3); \draw[ptr](c3)--(c4); \draw[ptr](c4)--(c5);
\node[keycell](ck) at (\xoff-1*\dx, 0){$b_2$}; \draw[hptr](ck)--(c1);
\node[endcell](ce) at (\xoff+5*\dx, 0){end}; \draw[ptr](c5)--(ce);

\node[pt, cz] at (\xoff+2*\dx, -0.40){};
\node[pt, cz] at (\xoff+3*\dx, -0.40){};
\node[lbl, cz, anchor=north] at (\xoff+2.5*\dx, -0.48){$P^T{:}\;(a_4,b_2,c_1),\;(a_5,b_2,c_1)$};

\node[pt, cs] at (\xoff+1*\dx, 0.40){};
\node[lbl, cs, anchor=south] at (\xoff+1*\dx, 0.48){$P^W{:}\;(a_2,b_2,d_1)$};
\end{tikzpicture}}
    \end{center}
    Lines~9--18: for each $S \subseteq \mC$, compute $V^S$ (dots, line~15), \textproc{Extend} to skip links $T^S$ (dashed, line~16), and \textproc{Merge} into $\bar{T}^S$ (solid, line~17). The three resulting skip links are the same as shown in Example~\ref{example:enumDS}. Here, we visualize the $S$-skip links separately for different $S$, but in the data structure they become one set of skip links. For $S = \{T,W\}$, we assume $R^* = T$. Extending $a_5$ finds no consistent incoming/outgoing skip links, giving $\textsf{prev}(a_5)\to\textsf{next}(a_5) = a_4\to a_6$. Extending $a_4$ finds an incoming skip link in $\bar{T}^{\{W\}}$ with label $d_1$, giving $a_1\to a_5$. Merging the two yields the skip link $a_1\to a_6$ with label $c_1 d_1$ in $\bar{T}^{\{T,W\}}$.
    \begin{center}
    \scalebox{0.8}{
\begin{tikzpicture}[
    cell/.style={draw=gray!65!black, fill=white, minimum width=0.9cm,
                 minimum height=0.44cm, font=\small, inner sep=2pt},
    keycell/.style={draw=pGray, dashed, fill=gray!10, minimum width=0.8cm,
                 minimum height=0.44cm, font=\small, inner sep=2pt},
    endcell/.style={draw=gray!65!black, fill=gray!12, minimum width=0.55cm,
                 minimum height=0.44cm, font=\scriptsize\sffamily, text=gray!55!black, inner sep=1pt},
    hptr/.style={->, >=stealth, thin, dashed, gray!50},
    ptr/.style={->, >=stealth, line width=0.5pt, gray!60!black},
    sk/.style={->, >=stealth, semithick, rounded corners=4pt},
    pt/.style={circle, fill, inner sep=1.5pt},
    lbl/.style={font=\scriptsize},
    steplbl/.style={font=\small\bfseries, anchor=west},
]
\colorlet{cs}{pBlue}
\colorlet{cz}{pOrange}
\colorlet{czw}{pPurple}

\def\dx{1.4}
\def\xoff{3.9}

\def\ya{0}
\node[steplbl] at (0, \ya) {$S{=}\{T\}$:};

\node[cell](d1) at (\xoff+0*\dx, \ya){$a_1$};
\node[cell](d2) at (\xoff+1*\dx, \ya){$a_2$};
\node[cell](d3) at (\xoff+2*\dx, \ya){$a_4$};
\node[cell](d4) at (\xoff+3*\dx, \ya){$a_5$};
\node[cell](d5) at (\xoff+4*\dx, \ya){$a_6$};
\draw[ptr](d1)--(d2); \draw[ptr](d2)--(d3); \draw[ptr](d3)--(d4); \draw[ptr](d4)--(d5);
\node[keycell](dk) at (\xoff-1*\dx, \ya){$b_2$}; \draw[hptr](dk)--(d1);
\node[endcell](de) at (\xoff+5*\dx, \ya){end}; \draw[ptr](d5)--(de);

\node[pt, cz] at (\xoff+2*\dx, \ya+0.46){};
\node[pt, cz] at (\xoff+3*\dx, \ya+0.46){};
\node[lbl, cz, anchor=east] at (\xoff+2*\dx-0.12, \ya+0.46){$V^{\{T\}}$};

\draw[sk,cz,dashed] (d2.south)--(\xoff+1*\dx, \ya-0.46)--(\xoff+3*\dx, \ya-0.46)--(d4.south);
\draw[sk,cz,dashed] (d3.south)--(\xoff+2*\dx, \ya-0.82)--(\xoff+4*\dx, \ya-0.82)--(d5.south);
\node[lbl,cz,anchor=west] at (\xoff+4*\dx+0.12, \ya-0.46){\textproc{Extend}};
\node[lbl,cz,anchor=west] at (\xoff+4*\dx+0.12, \ya-0.82){\textproc{Extend}};

\draw[sk,cz,thick] (d2.north)--(\xoff+1*\dx, \ya+0.72)--(\xoff+4*\dx, \ya+0.72)--(d5.north);
\node[lbl,cz,anchor=south] at (\xoff+2.5*\dx, \ya+0.78){$c_1$};
\node[lbl,cz,anchor=west] at (\xoff+4*\dx+0.12, \ya+0.72){\textproc{Merge}};

\def\yb{-2.0}
\node[steplbl] at (0, \yb) {$S{=}\{W\}$:};

\node[cell](e1) at (\xoff+0*\dx, \yb){$a_1$};
\node[cell](e2) at (\xoff+1*\dx, \yb){$a_2$};
\node[cell](e3) at (\xoff+2*\dx, \yb){$a_4$};
\node[cell](e4) at (\xoff+3*\dx, \yb){$a_5$};
\node[cell](e5) at (\xoff+4*\dx, \yb){$a_6$};
\draw[ptr](e1)--(e2); \draw[ptr](e2)--(e3); \draw[ptr](e3)--(e4); \draw[ptr](e4)--(e5);
\node[keycell](ek) at (\xoff-1*\dx, \yb){$b_2$}; \draw[hptr](ek)--(e1);
\node[endcell](ee) at (\xoff+5*\dx, \yb){end}; \draw[ptr](e5)--(ee);

\node[pt, cs] at (\xoff+1*\dx, \yb+0.46){};
\node[lbl, cs, anchor=east] at (\xoff+1*\dx-0.12, \yb+0.46){$V^{\{W\}}$};

\draw[sk,cs,dashed] (e1.south)--(\xoff+0*\dx, \yb-0.46)--(\xoff+2*\dx, \yb-0.46)--(e3.south);
\node[lbl,cs,anchor=west] at (\xoff+2*\dx+0.12, \yb-0.46){\textproc{Extend}};

\draw[sk,cs,thick] (e1.north)--(\xoff+0*\dx, \yb+0.72)--(\xoff+2*\dx, \yb+0.72)--(e3.north);
\node[lbl,cs,anchor=south] at (\xoff+1*\dx, \yb+0.78){$d_1$};
\node[lbl,cs,anchor=west] at (\xoff+2*\dx+0.12, \yb+0.72){\textproc{Merge}};

\def\yc{-3.9}
\node[steplbl] at (0, \yc) {$S{=}\{T,W\}$:};

\node[cell](f1) at (\xoff+0*\dx, \yc){$a_1$};
\node[cell](f2) at (\xoff+1*\dx, \yc){$a_2$};
\node[cell](f3) at (\xoff+2*\dx, \yc){$a_4$};
\node[cell](f4) at (\xoff+3*\dx, \yc){$a_5$};
\node[cell](f5) at (\xoff+4*\dx, \yc){$a_6$};
\draw[ptr](f1)--(f2); \draw[ptr](f2)--(f3); \draw[ptr](f3)--(f4); \draw[ptr](f4)--(f5);
\node[keycell](fk) at (\xoff-1*\dx, \yc){$b_2$}; \draw[hptr](fk)--(f1);
\node[endcell](fe) at (\xoff+5*\dx, \yc){end}; \draw[ptr](f5)--(fe);

\node[pt, czw] at (\xoff+2*\dx, \yc+0.46){};
\node[pt, czw] at (\xoff+3*\dx, \yc+0.46){};
\node[lbl, czw, anchor=east] at (\xoff+2*\dx-0.12, \yc+0.46){$V^{\{T,W\}}$};

\draw[sk,czw,dashed] (f1.south)--(\xoff+0*\dx, \yc-0.46)--(\xoff+3*\dx, \yc-0.46)--(f4.south);
\draw[sk,czw,dashed] (f3.south)--(\xoff+2*\dx, \yc-0.82)--(\xoff+4*\dx, \yc-0.82)--(f5.south);
\node[lbl,czw,anchor=west] at (\xoff+4*\dx+0.12, \yc-0.46){\textproc{Extend}};
\node[lbl,czw,anchor=west] at (\xoff+4*\dx+0.12, \yc-0.82){\textproc{Extend}};

\draw[sk,czw,thick] (f1.north)--(\xoff+0*\dx, \yc+0.72)--(\xoff+4*\dx, \yc+0.72)--(f5.north);
\node[lbl,czw,anchor=south] at (\xoff+2*\dx, \yc+0.78){$c_1 d_1$};
\node[lbl,czw,anchor=west] at (\xoff+4*\dx+0.12, \yc+0.72){\textproc{Merge}};

\path[use as bounding box] (-0.05, 1.02) rectangle (11.10, -5.05);

\end{tikzpicture}}
    \end{center}
\end{example}

\subsection{A Complete Example}

\label{example:complete-mixed-triangle}
    We give an overview of the complete algorithm on a triangle query:
    \[
        Q(x,y,z) \gets R(x) \wedge T(y,z) \wedge {\color{pBlue}\neg A(x,y)} \wedge {\color{pOrange}\neg B(x,z)}.
    \]
    We use the following instance:
    \begin{center}
{\small\renewcommand{\arraystretch}{1.1}%
\newcolumntype{Q}[1]{>{\centering\arraybackslash}p{#1}}%
\begin{tabular}[t]{@{}Q{0.7cm}@{}}
  \multicolumn{1}{c}{$R(x)$} \\ \toprule
  $x$ \\ \midrule
  $a_1$ \\ $a_2$ \\ $a_3$ \\ \bottomrule
\end{tabular}\qquad
\begin{tabular}[t]{@{}Q{0.7cm}Q{0.7cm}@{}}
  \multicolumn{2}{c}{$T(y,z)$} \\ \toprule
  $y$ & $z$ \\ \midrule
  $b_1$ & $c_1$ \\ $b_2$ & $c_1$ \\ $b_3$ & $c_2$ \\ \bottomrule
\end{tabular}\qquad
\begin{tabular}[t]{@{}Q{0.7cm}Q{0.7cm}@{}}
  \multicolumn{2}{c}{\textcolor{pBlue}{$\neg A(x,y)$}} \\ \toprule
  $x$ & $y$ \\ \midrule
  $a_1$ & $b_1$ \\ $a_1$ & $b_2$ \\ $a_2$ & $b_2$ \\ $a_3$ & $b_2$ \\ \bottomrule
\end{tabular}\qquad
\begin{tabular}[t]{@{}Q{0.7cm}Q{0.7cm}@{}}
  \multicolumn{2}{c}{\textcolor{pOrange}{$\neg B(x,z)$}} \\ \toprule
  $x$ & $z$ \\ \midrule
  $a_2$ & $c_1$ \\ $a_3$ & $c_1$ \\ $a_2$ & $c_2$ \\ \bottomrule
\end{tabular}}
    \end{center}

    We will use the variable order $\sigma=(x,y,z)$, so we start with the variable adjacent to the two negative edges ($x$). The tree decompositions for each $S\subseteq \mE^-$ corresponding to $\sigma$ are shown below:
    \begin{center}
    \begin{tikzpicture}[
    bag/.style={draw, rounded corners=3pt, inner xsep=3pt, inner ysep=2.5pt, font=\scriptsize, align=center},
    edge/.style={thick},
    slab/.style={font=\scriptsize, anchor=south},
]

\colorlet{c0}{pGray}
\colorlet{c1}{pBlue}
\colorlet{c2}{pOrange}
\colorlet{c3}{pPurple}

\def\yh{0.42}
\def\yx{0.00}
\def\yy{-0.74}
\def\yz{-1.48}

\node[slab, text=c0] at (0.00,\yh) {$S=\emptyset$};
\node[bag, fill=c0!7, draw=c0!75] (x0) at (0.00,\yx) {$\chi_\sigma^{\emptyset}(\underline{x})=\{\underline{x}\}$};
\node[bag, fill=c0!7, draw=c0!75] (y0) at (0.00,\yy) {$\chi_\sigma^{\emptyset}(\underline{y})=\{\underline{y},z\}$};
\node[bag, fill=c0!7, draw=c0!75] (z0) at (0.00,\yz) {$\chi_\sigma^{\emptyset}(\underline{z})=\{\underline{z}\}$};
\draw[edge, c0!75] (x0)--(y0);
\draw[edge, c0!75] (y0)--(z0);

\node[slab, text=c1] at (2.95,\yh) {$S=\{A\}$};
\node[bag, fill=c1!7, draw=c1!75] (x1) at (2.95,\yx) {$\chi_\sigma^{\{A\}}(\underline{x})=\{\underline{x},y\}$};
\node[bag, fill=c1!7, draw=c1!75] (y1) at (2.95,\yy) {$\chi_\sigma^{\{A\}}(\underline{y})=\{\underline{y},z\}$};
\node[bag, fill=c1!7, draw=c1!75] (z1) at (2.95,\yz) {$\chi_\sigma^{\{A\}}(\underline{z})=\{\underline{z}\}$};
\draw[edge, c1!75] (x1)--(y1);
\draw[edge, c1!75] (y1)--(z1);

\node[slab, text=c2] at (5.90,\yh) {$S=\{B\}$};
\node[bag, fill=c2!7, draw=c2!75] (x2) at (5.90,\yx) {$\chi_\sigma^{\{B\}}(\underline{x})=\{\underline{x},z\}$};
\node[bag, fill=c2!7, draw=c2!75] (y2) at (5.90,\yy) {$\chi_\sigma^{\{B\}}(\underline{y})=\{\underline{y},z\}$};
\node[bag, fill=c2!7, draw=c2!75] (z2) at (5.90,\yz) {$\chi_\sigma^{\{B\}}(\underline{z})=\{\underline{z}\}$};
\draw[edge, c2!75] (x2)--(y2);
\draw[edge, c2!75] (y2)--(z2);

\node[slab, text=c3] at (8.95,\yh) {$S=\{A,B\}$};
\node[bag, fill=c3!7, draw=c3!75] (x3) at (8.95,\yx) {$\chi_\sigma^{\{A,B\}}(\underline{x})=\{\underline{x},y,z\}$};
\node[bag, fill=c3!7, draw=c3!75] (y3) at (8.95,\yy) {$\chi_\sigma^{\{A,B\}}(\underline{y})=\{\underline{y},z\}$};
\node[bag, fill=c3!7, draw=c3!75] (z3) at (8.95,\yz) {$\chi_\sigma^{\{A,B\}}(\underline{z})=\{\underline{z}\}$};
\draw[edge, c3!75] (x3)--(y3);
\draw[edge, c3!75] (y3)--(z3);

\foreach \x in {1.475, 4.425, 7.40} {
  \draw[gray!30, line width=0.4pt] (\x, 0.80) -- (\x,-1.80);
}

\end{tikzpicture}
    \end{center}

    Running \textproc{BuildEnumDS} following $\sigma$ gives the following enumeration data structures.
    \begin{center}
    \resizebox{\textwidth}{!}{\input{figures/triangle-mixed-enumds.tikz}}
    \end{center}
    The rows show: the hypergraphs $\smH_1,\smH_2,\smH_3$, the data structures, and a table showing which skip links are created. In the preprocessing phase, we work left to right to build the data structures. In the enumeration phase, we work right to left, enumerating extensions to the next variable, and backtracking when we run out of values to extend.

    When we create $\mathbb{D}_x$, we first create a hash table with a unique key $\emptyset$, with a linked list with values $R(x)$. We then add skip links that skip values using $A(x,y)$ and $B(x,z)$. We also create merged skip links using both $A,B$, essentially by considering all $y,z \in T$ and checking if there are skip links from $A$ and $B$ with labels $y$ and $z$ which are non-containingly overlapping, which we merge into a skip link with label $y \cdot z$. Importantly, we do not merge skip links from $A$ and $B$ if the corresponding $y,z$ pair does not exist in $T$, since this would make the runtime $|I|^2$ instead of $|I|^{3/2}$. For example, no merged skip links are created for $b_1 c_2$ and $b_2 c_2$, since $(b_1,c_2), (b_2,c_2) \notin T$.

    Skip links that cover an \textit{entire} linked list (from bucket key to sentinel) are passed on to the next reduced query: their label and bucket become a tuple of a new negative atom. In $\mathbb{D}_x$, the $\{A\}$-skip link with label $b_2$ covers the whole list, so ${\color{pBlue}U_1^{A}} = \{b_2\}$. Similarly, the merged skip link with label $b_1 c_1$ gives ${\color{pPurple}U_1^{A,B}} = \{(b_1,c_1)\}$, while ${\color{pOrange}U_1^{B}} = \emptyset$. 
    
    Repeating for $y$, in bucket $c_1$ of $\mathbb{D}_y$, the $\{U_1^{A,B}\}$-skip link skips $b_1$, since $(b_1,c_1) \in U_1^{A,B}$, and the $\{U_1^A\}$-skip link skips $b_2$, since $b_2 \in U_1^A$. Neither covers the whole list on its own, so neither is passed up, i.e., $U_2^{U_1^A} = U_2^{U_1^{A,B}} = \emptyset$. However, they overlap, so they are merged into a $\{U_1^A, U_1^{A,B}\}$-skip link from $c_1$ to the sentinel, which covers the whole list. Hence $c_1$ is passed up, giving $U_2^{U_1^A,U_1^{A,B}} = \{c_1\}$, while $U_1^{B}$ is carried over. Hence $\mathbb{D}_z$ skips $c_1$.

    In the enumeration phase, we follow the data-structures from right to left. This gives the output tuples $(a_1,b_3,c_2)$ and $(a_3,b_3,c_2)$.
\subsection{Algorithm Analysis}

We will describe the key ideas of the algorithm analysis and present the key lemmas. The full analysis and proofs are quite long and thus deferred to \apporfull{sec:app:algorithm-analysis}. The analysis consists of three parts: (1) proving correctness of the algorithm, (2) proving the runtime of the preprocessing phase, and (3) proving that the enumeration phase runs with constant delay.

\paragraph{Correctness} The key ideas behind correctness are the following: First, lines~$2$--$3$ create a hash table which allows us to, given a tuple over suffix variables, correctly find all extensions to the next variable, given positive atoms. It then just remains to ``filter'' out the extensions which are blocked by a negative hyperedge. We do this by arguing that for any consecutive set of nodes in the linked list that are all blocked by some negative hyperedges and are consistent with the same suffix tuple, there is a skip link that covers all of them.

\begin{restatable}{lemma}{reductioncorrectlemma}
\label{lemma:ReductionCorrect}
\label{lemma:reductioncorrect}
    Let $Q$ be the query over $\smH$ with input $I$, and for each $i\geq 1$ let $Q_i$ over $\smH_i=(V_i,\mE_i^+,\mE_i^-)$ with input $I_i$ be the query obtained after $i-1$ successive applications of the reduction of Algorithm~\ref{algorithm:BuildEnumDS}. Then $Q_i(I_i) = \pi_{V_i}\, Q(I)$.
\end{restatable}

\paragraph{Preprocessing Runtime} The runtime of the preprocessing phase is dominated by joins. The crucial insight is that we can, for each join, analyze it as a Conjunctive Query over a bag $\chi_\sigma^S(\sigma_i)$, for some set $S$. This means that the tree decompositions $\mT_\sigma^S$ collectively capture all join results, bounding their runtime by $O(|I|^{\sfhow(\smH)})$.

\begin{restatable}{lemma}{enumdsconstructionlemma}
\label{lemma:EnumDSConstruction}
    Given a signed hypergraph $\smH$ with input $I$, Algorithm~\ref{algorithm:BuildEnumDS} runs in
    $ O( |I|^{\sfhow(\smH)} ) $ time.
\end{restatable}

\paragraph{Constant Delay Enumeration} Constant-delay enumeration is proved together with the correctness of the algorithm. The key idea is that enumeration just follows the linked lists or skip links. By the correctness argument, we always need to follow \textit{at most} one skip link before reaching the next value to output over a variable. Finally, since the reduced queries contain relations $U_i^S$ which are ``passed up'' from earlier reduced queries, we can ensure we only enumerate the necessary suffixes.

\begin{restatable}{lemma}{enumfunctionstepoutputslemma}
\label{lemma:enumfunctionstepoutputs}
    Let $i\in[n]$, and let $t_{i+1}$ be an output over $\smH_{i+1}$ (recall that $\smH_{n+1}$ has no vertices, so its only possible output is the empty tuple $()$). Then, \textproc{Enumerate} iterates with constant delay over all outputs $t_i \sim t_{i+1}$ over $\smH_i$.
\end{restatable}

\section{Improving Time Bounds with Multiple Variable Orders}
\label{sec:multi-td}

So far our algorithms are based on a single variable order. In this section, we discuss how using multiple variable orders can improve the runtime bounds. This is analogous to Conjunctive Query evaluation, where the runtime can be improved from fractional hypertree width to submodular width by utilizing multiple tree decompositions. As a proof of concept, we apply this idea to give an improved runtime bound for $k$-cycles with negation, improving the preprocessing time from $O(|I|^2)$ to $O(|I|^{3/2})$. We leave it as an open question to give a more general characterization of the power of multiple variable orders for $\cqneg$ queries.

\subsection{Cycle Queries}

Consider the following general family of $k$-cycle queries:
\[
C_k^{\pm}(x_1, \dots, x_k) \gets X_1(x_1) \wedge \dots \wedge X_k(x_k) \wedge \pm R_1(x_1,x_2) \wedge \dots \wedge \pm R_k(x_k,x_1).
\]

Here, the relations $R_1,\dots,R_k$ can be either positive or negative. Let $m$ be the maximum size of $R_1,\dots,R_k$. If all edges are positive, then we obtain
\[
C_k(x_1, \dots, x_k) \gets X_1(x_1) \wedge \dots \wedge X_k(x_k) \wedge R_1(x_1,x_2) \wedge \dots \wedge R_k(x_k,x_1).
\]
which can be evaluated in time $O(m^{\textsf{subw}})$, where $\textsf{subw}(C_k) = 2 - \frac{1}{\lceil k/2 \rceil}$. Hence, as $k$ grows, the runtime approaches $O(m^2)$.
Next, we consider the case where at least one edge is negative.
Before we present the full algorithm for $C_k^{\pm}$, we first provide some intuition of the main technical idea.

\paragraph{Intuition.} Suppose we want to compute $C_k$ (all positive edges) using the variable order $x_1, \dots, x_k$, with canonical decomposition $\{x_1, x_2, x_k\} - \{x_2, x_3, x_k\} - \dots - \{x_{k-2}, x_{k-1}, x_k\} - \{x_{k-1},x_k\} - \{x_{k}\}$. The size of each bag is bounded by $O(m^2)$, so the total runtime for evaluating $C_k$ is $O(m^2)$. This bound can be improved by using multiple variable orders. Given a relation $R(x,y)$, we define the degree of $x$ given $y$ as $\deg_R(x \mid y) := \max_{v \in \pi_y R} |\{u : (u,v) \in R\}|$. The key idea is to partition the relations $R_1, R_2, \dots, R_k$ into subrelations with bounded degrees. Indeed, suppose that for every $i$, $\deg_{R_i}(x_{i} \mid x_{i+1}) \leq d$. Then, one can see that the size of each bag of the form $\{x_{i}, x_{i+1}, x_{k}\}$ is bounded by $O(m d^{i})$ instead of $O(m^2)$, since every increment of $i$ expands the number of values of $x_k$ reachable from a value over $x_i$ by a factor of at most $d$. By choosing an appropriate parameter $d$ and appropriate variable orders, we can achieve the improved runtime bound for $C_k$.

The expansion by a factor of $d$ is inherent in the all-positive case and this is what causes the runtime dependence on $k$. However, when we have at least one negative edge, we can circumvent this obstacle because, instead of storing for each value over $x_i$ the set of values of $x_1$ that it connects to through a path of positive edges, we store the complement (via the skip-link data structure). In this case, as the key formal lemma below shows, we pay no penalty of expansion.

We will say that a tuple $t$ \textit{satisfies} an atom $R$, if either $R$ is positive and $t$ is consistent with a tuple in $R$, or $R$ is negative and $t$ is not consistent with any tuple in $R$.

\begin{lemma}
\label{lemma:keykcyclelemma}
    Let $k \geq 2$ and consider a $k$-path query of the form:
    \[
    P_k^{\pm}(x_1, \dots, x_k) \gets X_1(x_1) \wedge \dots \wedge X_k(x_k) \wedge \neg R_1(x_1,x_2) \wedge \dots \wedge \pm R_{k-1}(x_{k-1},x_k).
    \]
    Here, the first atom is negative, and the remaining atoms are either positive or negative. Suppose that the following conditions hold for an instance $I$:
    \begin{enumerate}
        \item $R_1$ has the degree bound: $\deg_{R_1}(x_1 \mid x_2) \leq d$.
        \item for any $i\geq 3$, and any value $v_i$ over $X_i(x_i)$, there exists $v_{i-1}$ over $X_{i-1}(x_{i-1})$ such that $(v_{i-1},v_i)$ satisfies $R_{i-1}$.
    \end{enumerate}
    Then, for any value $v_k$ over $X_k(x_k)$, there are at most $d$ values $v_1$ over $X_1(x_1)$ such that there is no output tuple $t \in P_k^{\pm}(I)$ with $t.x_1 = v_1$ and $t.x_k = v_k$.
\end{lemma}
\begin{proof}
    We prove this by induction over $k$. In the base case $k=2$, the statement follows from condition $(1)$.

    Next, let $v_k$ be a value over $X_k(x_k)$. Let $V_{k-1}$ be the set of values over $X_{k-1}(x_{k-1})$ such that for $v_{k-1}\in V_{k-1}$, $(v_{k-1},v_k)$ satisfies $R_{k-1}$. $V_{k-1}$ is non-empty by $(2)$. Let $P_{k-1}^\pm$ be the subquery over the same input where atoms $X_k$ and $R_{k-1}$ are removed.

    We claim that a value $v_1$ over $X_1(x_1)$ does not occur in an output tuple $t$ with $v_k$ in $P_k^\pm$ if and only if $v_1$ does not occur in an output tuple $t'$ of $P_{k-1}^\pm$ with any $v_{k-1}\in V_{k-1}$.
    $\Rightarrow$: There cannot be an output tuple $t'$ in $P_{k-1}^\pm$ with $t'.{x_1} = v_1$ and $t'.{x_{k-1}}\in V_{k-1}$, because then we can extend $t'$ to an output tuple $t$ of $P_k^\pm$ with $t.x_1 = v_1$ and $t.x_k = v_k$, since $(t'.{x_{k-1}},v_k)$ satisfies $R_{k-1}$.
    $\Leftarrow$: If $v_1$ is not output with any $v_{k-1}\in V_{k-1}$ in $P_{k-1}^\pm$, then any output tuple $t$ over $P_k^\pm$ must have value over $X_{k-1}(x_{k-1})$ that is not in $V_{k-1}$. But then it cannot have $t.x_{k}=v_k$ since it does not satisfy $R_{k-1}$.

    In other words, for a $v_1$ to not be output with $v_k$ in $P_k^\pm$, it must not be output with any $v_{k-1}\in V_{k-1}$ in $P_{k-1}^\pm$. By the induction hypothesis, for each $v_{k-1}\in V_{k-1}$, there are at most $d$ such values.
\end{proof}

An important intuition in the above lemma is that we have switched from an \textit{or} operation to an \textit{and} operation when we grow to the next variable.

\paragraph{Algorithm.}
Given a query in the family $C_k^{\pm}$ with at least one negative edge (say $R_k$) and an instance $I$ with relation cardinalities at most $m$, we will first partition $R_k$ into two subrelations $R_k^l$ and $R_k^h$ such that $\deg_{R_k^l}(x_k \mid x_1) \leq \sqrt{m}$ and $\deg_{R_k^h}(x_1 \mid x_k) \leq \sqrt{m}$ -- this can be done by sending all values of $x_1$ with degree at most $\sqrt{m}$ to $R_k^l$, and the remaining values to $R_k^h$. Let $X_1^l$ be the values of $X_1$ with degree at most $\sqrt{m}$ in $R_k$ (including those that do not occur in $R_k$), and $X_1^h = X_1 \setminus X_1^l$. We will then create two subinstances of $I$: $I^l$ is obtained from $I$ by replacing $R_k$ with $R_k^l$ and $X_1$ with $X_1^l$, and $I^h$ is obtained from $I$ by replacing $R_k$ with $R_k^h$ and $X_1$ with $X_1^h$. For each subinstance, we will perform a simple preprocessing step to remove tuples such that condition $(2)$ of Lemma~\ref{lemma:keykcyclelemma} is satisfied: this can be easily achieved in linear time $O(m)$ via a sequence of semijoins/antijoins.

Then, we will use the variable order $x_1, x_2, \dots, x_{k-1}, x_k$ to evaluate $C_k^{\pm}$ on $I^l$, and the variable order $x_k, x_{k-1}, \dots, x_2, x_1$ for $I^h$. For each subinstance, we will apply the algorithm essentially as presented in the previous section, with only minor modifications to handle a degree constraint, which we did not consider before. The key observation is that we can use Lemma~\ref{lemma:keykcyclelemma} to improve the bound of the data structure from $O(m^2)$ to $O(m\sqrt{m})$, since the reduced query has an atom with a degree constraint of $\sqrt{m}$. We therefore achieve preprocessing time $O(m^{3/2})$.

\begin{restatable}{theorem}{kcycletheorem}\label{theorem:kcycle}
    For every $k \geq 2$, if the query $C_k^{\pm}$ contains at least one negative edge, then it can be enumerated with constant delay and preprocessing time $O(|I|^{3/2})$.
\end{restatable}

The proof is in \apporfull{sec:app:proofkcycle}. The figure below shows the partition of $R_k$ and the variable orders used for each subinstance.

\begin{example}
    As an example, consider the query $Q(x_1,x_2,x_3,x_4) \gets R_1(x_1,x_2) \wedge R_2(x_2,x_3) \wedge R_3(x_3,x_4) \wedge \neg R_4(x_4,x_1)$.
    \begin{center}
    \scalebox{0.8}{\begin{tikzpicture}[
    scale=0.9,
    vtx/.style={circle, draw, inner sep=2pt, font=\small},
    pedge/.style={draw=black!70, thick},
    negedge/.style={dashed, draw=pOrange, very thick},
    bag/.style={draw, rounded corners=3pt, inner sep=4pt, font=\scriptsize,
                fill=gray!8, draw=pGray!60},
    rbag/.style={draw, rounded corners=3pt, inner sep=4pt, font=\scriptsize,
                fill=pOrange!10, draw=pOrange!70},
    fwd/.style={draw=gray!60, thick},
]

\begin{scope}[xshift=0cm]
    \node[vtx] (n1) at ( 0.0,  0.9) {$x_1$};
    \node[vtx] (n2) at ( 1.1,  0.0) {$x_2$};
    \node[vtx] (n3) at ( 0.0, -0.9) {$x_3$};
    \node[vtx] (n4) at (-1.1,  0.0) {$x_4$};
    \draw[pedge] (n1)--(n2) node[midway, right, font=\scriptsize]{$R_1$};
    \draw[pedge] (n2)--(n3) node[midway, right, font=\scriptsize]{$R_2$};
    \draw[pedge] (n3)--(n4) node[midway, left,  font=\scriptsize]{$R_3$};
    \draw[negedge] (n4)--(n1) node[midway, above left, inner sep=1pt, font=\scriptsize]{$\neg R_4$};
\end{scope}

\draw[gray!25, thick] (1.85, 1.40) -- (1.85, -2.75);

\node[font=\scriptsize, text=black!80]
    at (5.1, 1.32) {$\sigma=(x_1,\ldots,x_4),\;\deg_{R_4}(x_4\mid x_1)\leq\sqrt{m}$};

\begin{scope}[xshift=3.7cm]
    \node[font=\scriptsize] at (0, 0.8) {$S=\emptyset$};
    \node[bag] (a4) at (0,  0.2) {$\{\underline{x_4}\}$};
    \node[bag] (a3) at (0, -0.65) {$\{\underline{x_3},x_4\}$};
    \node[bag] (a2) at (0, -1.5) {$\{\underline{x_2},x_3\}$};
    \node[bag] (a1) at (0, -2.35) {$\{\underline{x_1},x_2\}$};
    \draw[fwd] (a4)--(a3); \draw[fwd] (a3)--(a2); \draw[fwd] (a2)--(a1);
\end{scope}

\begin{scope}[xshift=6.3cm]
    \node[font=\scriptsize] at (0, 0.8) {$S=\{R_4\}$};
    \node[rbag] (b4) at (0,  0.2) {$\{\underline{x_4}\}$};
    \node[rbag] (b3) at (0, -0.65) {$\{\underline{x_3},x_4\}$};
    \node[rbag] (b2) at (0, -1.5) {$\{\underline{x_2},x_3,x_4\}$};
    \node[rbag] (b1) at (0, -2.35) {$\{\underline{x_1},x_2,x_4\}$};
    \draw[fwd] (b4)--(b3); \draw[fwd] (b3)--(b2); \draw[fwd] (b2)--(b1);
\end{scope}

\draw[gray!25, thick] (8.45, 1.40) -- (8.45, -2.75);

\node[font=\scriptsize, text=black!80]
    at (11.6, 1.32) {$\sigma=(x_4,\ldots,x_1),\;\deg_{R_4}(x_1\mid x_4)\leq\sqrt{m}$};

\begin{scope}[xshift=9.8cm]
    \node[font=\scriptsize] at (0, 0.8) {$S=\emptyset$};
    \node[bag] (c1) at (0,  0.2) {$\{\underline{x_1}\}$};
    \node[bag] (c2) at (0, -0.65) {$\{\underline{x_2},x_1\}$};
    \node[bag] (c3) at (0, -1.5) {$\{\underline{x_3},x_2\}$};
    \node[bag] (c4) at (0, -2.35) {$\{\underline{x_4},x_3\}$};
    \draw[fwd] (c1)--(c2); \draw[fwd] (c2)--(c3); \draw[fwd] (c3)--(c4);
\end{scope}

\begin{scope}[xshift=12.4cm]
    \node[font=\scriptsize] at (0, 0.8) {$S=\{R_4\}$};
    \node[rbag] (d1) at (0,  0.2) {$\{\underline{x_1}\}$};
    \node[rbag] (d2) at (0, -0.65) {$\{\underline{x_2},x_1\}$};
    \node[rbag] (d3) at (0, -1.5) {$\{\underline{x_3},x_2,x_1\}$};
    \node[rbag] (d4) at (0, -2.35) {$\{\underline{x_4},x_3,x_1\}$};
    \draw[fwd] (d1)--(d2); \draw[fwd] (d2)--(d3); \draw[fwd] (d3)--(d4);
\end{scope}

\end{tikzpicture}}
    \end{center}
    For the first variable order $(x_1, x_2, x_3, x_4)$, in the bag where $x_i$ disappears (underlined in the figure), the negative atom $R(x_i,x_4)$ has degree constraint $\deg_R(x_4 \mid x_i) \leq \sqrt{m}$.
\end{example}

\subsection{Generalizing the Width Measure}

Let $\mH_k^+$ be the hypergraph of the $k$-cycle query where all edges are positive, and $\smH_k$ a signed hypergraph of some $k$-cycle query, where at least one edge is negative. For $k \geq 3$, we get the following inequality for the known width measures:
{\setlength{\abovedisplayskip}{6pt}\setlength{\belowdisplayskip}{6pt}%
\[
\underbrace{\textsf{width}(\smH_k)}_{\frac{3}{2}}
\leq \underbrace{\textsf{subw}(\mH_k^+)}_{2 - \frac{1}{\lceil k/2 \rceil}}
\leq \underbrace{\textsf{fhw}(\mH_k^+)}_{\frac{3}{2}\text{ if }k=3,\ 2\text{ if }k\geq 4}
\leq \underbrace{\sfhow(\smH_k)}_{\substack{\frac{3}{2}\text{ if }k=3\text{ and some }R_i\text{ is positive},\\ 2\text{ otherwise}}}
\]}

Here, $\textsf{width}(\smH_k)$ is the width measure with multiple variable orders, obtained by Theorem~\ref{theorem:kcycle}. We leave it as an open question how one can generalize the formula for $\sfhow$ to capture the power of multiple variable orders, similar to how submodular width generalizes fractional hypertree width for CQs. As one can see from the above inequalities, when negation is present, the optimal width measure may behave very differently from submodular width.

\section{Related Work}

\paragraph{Width Measures for Queries.} For positive CQs, increasingly fine-grained tractability is captured by a hierarchy of hypergraph width measures: treewidth and hypertree width~\cite{GottlobLeoneScarcello02}, fractional hypertree width~\cite{GroheMarx06}, and submodular width~\cite{Marx10}. Generalized hypertree width and related hypergraph invariants are studied in~\cite{AdlerGottlobGrohe07,GottlobPichler01}. On the algorithmic side, the submodular-width preprocessing bound is realized by the \textsf{PANDA} algorithm~\cite{Panda,PandaExpress}.

\paragraph{Enumeration for Conjunctive Queries.} Bagan, Durand, and Grandjean~\cite{BaganDG07} showed that free-connex acyclic CQs admit constant-delay enumeration with linear-time preprocessing, a baseline that has since been extended in several directions: dichotomies for unions of CQs~\cite{CarmeliKroll21}, factorized representations as a succinct intermediate format~\cite{OlteanuZavodny15}, and enumeration under updates via dynamic-Yannakakis-style algorithms~\cite{IdrisUgarteVansummeren17,BerkholzKeppelerSchweikardt17}. For $\cqneg$, Zhao et al.~\cite{SignedAcyclic} identified the free-connex signed-acyclic class as the precise boundary for linear-time preprocessing with constant delay, sharpening earlier polylogarithmic-delay results of Brault-Baron~\cite{BraultBaron12,BraultBaronThesis}.

\paragraph{Space--Time Trade-offs and Parameterized Delay.} Beyond the constant-delay regime, a more recent line of work parameterizes the trade-off between preprocessing cost and per-answer latency. Deep and Koutris~\cite{DeepKoutris18} introduce compressed representations of CQ results that explicitly expose this trade-off, building on the factorized representations of~\cite{OlteanuZavodny15}. From a different angle, knowledge-compilation approaches~\cite{AmarilliCapelliMonetSenellart20} relate width parameters to the size of d-DNNF and decision-DNNF circuits encoding the answer set, yielding a circuit-based measure of the same trade-off.

\paragraph{Conjunctive Queries with Negation.} Brault-Baron~\cite{BraultBaron12,BraultBaronThesis} initiated the structural study of $\cqneg$, proving that all-negative queries are tractable if and only if their hypergraph is $\beta$-acyclic, and introduced the signed-acyclicity property used throughout this paper. Going beyond $\beta$-acyclicity, Lanzinger~\cite{NestSetWidth} introduced nest-set width as a finer hypergraph parameter that captures tractability of the decision problem for a strictly larger class of queries. \rev{Capelli et al.~\cite{CapelliIrwin24,CapelliCarmeliIrwinSalvati26} showed that, for queries with only negative atoms, bounded nest-set width also makes counting tractable, a question left open in~\cite{NestSetWidth}.} Related structural parameters have been studied in adjacent settings: Gottlob and Pichler~\cite{GottlobPichler01} compare hypergraph acyclicity and hypertree width to clique-width in the context of model checking, and Carbonnel, Romero, and \v{Z}ivn\'y~\cite{CarbonnelRomeroZivny20} introduce \emph{point-width} for Max-CSPs as a generalization of $\beta$-acyclicity.

\paragraph{\rev{Direct Access for $\cqneg$.}}
Direct-access algorithms support random access to the $j$-th answer in time independent of the database size. 
\rev{Capelli et al.~\cite{CapelliIrwin24,CapelliCarmeliIrwinSalvati26} study direct access for full $\cqneg$ queries, where the answers are ordered lexicographically according to a variable order, and give two algorithms whose preprocessing time is governed by $\sfhow$. The first is a reduction to positive queries: direct access for $Q$ is simulated by binary search over direct-access structures for the positive queries over the hypergraphs $\mH^S$, which are built using the algorithm of Bringmann, Carmeli, and Mengel~\cite{BringmannCarmeliMengel22}. The second compiles $Q(I)$ into an ordered circuit, which branches on one variable at a time and caches subproblems. After annotating the circuit with counts, direct access is answered by walking down the circuit. The first algorithm has preprocessing time $O(|I|^{\sfhow})$ and the second $\tilde{O}(|I|^{\sfhow})$, and both have polylogarithmic access time.
The authors also show that for self-join-free queries, their preprocessing time is optimal for direct access in lexicographic order, assuming the Zero-Clique Conjecture. 
These results naturally give polylogarithmic-delay enumeration algorithms for full $\cqneg$ queries.
Our algorithm extends their results by keeping the preprocessing time $O(|I|^{\sfhow})$ but reducing the delay to constant, by using the linked-list and skip-link framework from \cite{SignedAcyclic}.}

\section{Conclusion}

We have presented an algorithm for enumerating the output of $\cqneg$ queries with constant delay, where the preprocessing time is described by the \rev{fractional signed hyperorder width}. While this deepens the understanding of the complexity of computing $\cqneg$ queries, understanding the power of multiple variable orders is an enticing direction for future research.

\bibliographystyle{ACM-Reference-Format}
\bibliography{references}

@article{SignedAcyclic,
  author    = {Hangdong Zhao and Austen Z. Fan and Xiating Ouyang and Paraschos Koutris},
  title     = {Conjunctive Queries with Negation and Aggregation: {A} Linear Time Characterization},
  journal   = {Proceedings of the {ACM} on Management of Data ({PACMMOD})},
  volume    = {2},
  number    = {2},
  pages     = {97:1--97:26},
  year      = {2024},
  publisher = {ACM},
  doi       = {10.1145/3651138},
  note      = {PODS 2024}
}

@inproceedings{FAQ,
  author    = {Mahmoud Abo Khamis and Hung Q. Ngo and Atri Rudra},
  title     = {{FAQ}: Questions Asked Frequently},
  booktitle = {Proceedings of the 35th {ACM} {SIGMOD}-{SIGACT}-{SIGAI} Symposium on Principles of Database Systems ({PODS})},
  pages     = {13--28},
  year      = {2016},
  publisher = {ACM},
  doi       = {10.1145/2902251.2902280}
}

@inproceedings{Panda,
  author    = {Mahmoud Abo Khamis and Hung Q. Ngo and Dan Suciu},
  title     = {What Do {S}hannon-type Inequalities, Submodular Width, and Disjunctive {D}atalog Have to Do with One Another?},
  booktitle = {Proceedings of the 36th {ACM} {SIGMOD}-{SIGACT}-{SIGAI} Symposium on Principles of Database Systems ({PODS})},
  pages     = {429--444},
  year      = {2017},
  publisher = {ACM},
  doi       = {10.1145/3034786.3056105}
}

@inproceedings{PandaExpress,
  author    = {Mahmoud Abo Khamis and Hung Q. Ngo and Dan Suciu},
  title     = {{PANDAExpress}: a Simpler and Faster {PANDA} Algorithm},
  booktitle = {Proceedings of the 45th {ACM} {SIGMOD}-{SIGACT}-{SIGAI} Symposium on Principles of Database Systems ({PODS})},
  year      = {2026},
  publisher = {ACM},
  eprint    = {2512.10217},
  archivePrefix = {arXiv}
}

@inproceedings{BaganDG07,
  author    = {Guillaume Bagan and Arnaud Durand and Etienne Grandjean},
  title     = {On Acyclic Conjunctive Queries and Constant Delay Enumeration},
  booktitle = {Proceedings of the 21st International Workshop on Computer Science Logic ({CSL})},
  pages     = {208--222},
  year      = {2007},
  publisher = {Springer},
  doi       = {10.1007/978-3-540-74915-8\_18}
}

@article{NPRR,
  author    = {Hung Q. Ngo and Ely Porat and Christopher R{\'{e}} and Atri Rudra},
  title     = {Worst-case Optimal Join Algorithms},
  journal   = {Journal of the {ACM}},
  volume    = {65},
  number    = {3},
  pages     = {16:1--16:40},
  year      = {2018},
  publisher = {ACM},
  doi       = {10.1145/3180143}
}

@inproceedings{LeapfrogTriejoin,
  author    = {Todd L. Veldhuizen},
  title     = {Leapfrog Triejoin: A Simple, Worst-Case Optimal Join Algorithm},
  booktitle = {Proceedings of the 17th International Conference on Database Theory ({ICDT})},
  pages     = {96--106},
  year      = {2014},
  doi       = {10.5441/002/icdt.2014.13}
}

@inproceedings{BraultBaron12,
  author    = {Johann Brault-Baron},
  title     = {A Negative Conjunctive Query is Easy if and only if it is Beta-Acyclic},
  booktitle = {Proceedings of the 26th International Workshop / 21st Annual Conference of the {EACSL} on Computer Science Logic ({CSL})},
  series    = {LIPIcs},
  volume    = {16},
  pages     = {137--151},
  year      = {2012},
  publisher = {Schloss Dagstuhl -- Leibniz-Zentrum f{\"u}r Informatik},
  doi       = {10.4230/LIPIcs.CSL.2012.137}
}

@phdthesis{BraultBaronThesis,
  author    = {Johann Brault-Baron},
  title     = {De la pertinence de l'{\'e}num{\'e}ration: complexit{\'e} en logiques propositionnelle et du premier ordre},
  school    = {Universit{\'e} de Caen Basse-Normandie},
  year      = {2013}
}

@inproceedings{GottlobPichler01,
  author    = {Georg Gottlob and Reinhard Pichler},
  title     = {Hypergraphs in Model Checking: Acyclicity and Hypertree-Width versus Clique-Width},
  booktitle = {Proceedings of the 28th International Colloquium on Automata, Languages and Programming ({ICALP})},
  series    = {Lecture Notes in Computer Science},
  volume    = {2076},
  pages     = {708--719},
  year      = {2001},
  publisher = {Springer},
  doi       = {10.1007/3-540-48224-5\_58}
}

@inproceedings{NestSetWidth,
  author    = {Matthias Lanzinger},
  title     = {Tractability Beyond {$\beta$}-Acyclicity for Conjunctive Queries with Negation},
  booktitle = {Proceedings of the 40th {ACM} {SIGMOD}-{SIGACT}-{SIGAI} Symposium on Principles of Database Systems ({PODS})},
  pages     = {355--369},
  year      = {2021},
  publisher = {ACM},
  doi       = {10.1145/3452021.3458308}
}

@inproceedings{Marx10,
  author    = {D{\'a}niel Marx},
  title     = {Tractable Hypergraph Properties for Constraint Satisfaction and Conjunctive Queries},
  booktitle = {Proceedings of the 42nd {ACM} Symposium on Theory of Computing ({STOC})},
  pages     = {735--744},
  year      = {2010},
  publisher = {ACM},
  doi       = {10.1145/1806689.1806790}
}

@inproceedings{GroheMarx06,
  author    = {Martin Grohe and D{\'a}niel Marx},
  title     = {Constraint Solving via Fractional Edge Covers},
  booktitle = {Proceedings of the 17th Annual {ACM-SIAM} Symposium on Discrete Algorithms ({SODA})},
  pages     = {289--298},
  year      = {2006},
  publisher = {SIAM},
  doi       = {10.1145/1109557.1109590}
}

@article{GottlobLeoneScarcello02,
  author    = {Georg Gottlob and Nicola Leone and Francesco Scarcello},
  title     = {Hypertree Decompositions and Tractable Queries},
  journal   = {Journal of Computer and System Sciences},
  volume    = {64},
  number    = {3},
  pages     = {579--627},
  year      = {2002},
  publisher = {Elsevier},
  doi       = {10.1006/jcss.2001.1809}
}

@article{AdlerGottlobGrohe07,
  author    = {Isolde Adler and Georg Gottlob and Martin Grohe},
  title     = {Hypertree Width and Related Hypergraph Invariants},
  journal   = {European Journal of Combinatorics},
  volume    = {28},
  number    = {8},
  pages     = {2167--2181},
  year      = {2007},
  publisher = {Elsevier},
  doi       = {10.1016/j.ejc.2007.04.013}
}

@article{CarmeliKroll21,
  author    = {Nofar Carmeli and Markus Kr{\"o}ll},
  title     = {On the Enumeration Complexity of Unions of Conjunctive Queries},
  journal   = {{ACM} Transactions on Database Systems},
  volume    = {46},
  number    = {2},
  pages     = {5:1--5:41},
  year      = {2021},
  publisher = {ACM},
  doi       = {10.1145/3450263}
}

@inproceedings{IdrisUgarteVansummeren17,
  author    = {Muhammad Idris and Mart{\'\i}n Ugarte and Stijn Vansummeren},
  title     = {The Dynamic {Y}annakakis Algorithm: Compact and Efficient Query Processing Under Updates},
  booktitle = {Proceedings of the 2017 {ACM} International Conference on Management of Data ({SIGMOD})},
  pages     = {1259--1274},
  year      = {2017},
  publisher = {ACM},
  doi       = {10.1145/3035918.3064027}
}

@inproceedings{BerkholzKeppelerSchweikardt17,
  author    = {Christoph Berkholz and Jens Keppeler and Nicole Schweikardt},
  title     = {Answering Conjunctive Queries under Updates},
  booktitle = {Proceedings of the 36th {ACM} {SIGMOD}-{SIGACT}-{SIGAI} Symposium on Principles of Database Systems ({PODS})},
  pages     = {303--318},
  year      = {2017},
  publisher = {ACM},
  doi       = {10.1145/3034786.3034789}
}

@article{OlteanuZavodny15,
  author    = {Dan Olteanu and Jakub Z{\'a}vodn{\'y}},
  title     = {Size Bounds for Factorised Representations of Query Results},
  journal   = {{ACM} Transactions on Database Systems},
  volume    = {40},
  number    = {1},
  pages     = {2:1--2:44},
  year      = {2015},
  publisher = {ACM},
  doi       = {10.1145/2656335}
}

@inproceedings{DeepKoutris18,
  author    = {Shaleen Deep and Paraschos Koutris},
  title     = {Compressed Representations of Conjunctive Query Results},
  booktitle = {Proceedings of the 37th {ACM} {SIGMOD}-{SIGACT}-{SIGAI} Symposium on Principles of Database Systems ({PODS})},
  pages     = {307--322},
  year      = {2018},
  publisher = {ACM},
  doi       = {10.1145/3196959.3196979}
}

@article{AmarilliCapelliMonetSenellart20,
  author    = {Antoine Amarilli and Florent Capelli and Mika{\"e}l Monet and Pierre Senellart},
  title     = {Connecting Knowledge Compilation Classes and Width Parameters},
  journal   = {Theory of Computing Systems},
  volume    = {64},
  number    = {5},
  pages     = {861--914},
  year      = {2020},
  publisher = {Springer},
  doi       = {10.1007/s00224-019-09930-4}
}

@article{CarbonnelRomeroZivny20,
  author    = {Cl{\'e}ment Carbonnel and Miguel Romero and Stanislav {\v Z}ivn{\'y}},
  title     = {Point-Width and {M}ax-{CSP}s},
  journal   = {{ACM} Transactions on Algorithms},
  volume    = {16},
  number    = {4},
  pages     = {54:1--54:28},
  year      = {2020},
  publisher = {ACM},
  doi       = {10.1145/3409447}
}

@inproceedings{CapelliIrwin24,
  author    = {Florent Capelli and Oliver Irwin},
  title     = {Direct Access for Conjunctive Queries with Negations},
  booktitle = {Proceedings of the 27th International Conference on Database Theory ({ICDT})},
  series    = {LIPIcs},
  volume    = {290},
  pages     = {13:1--13:20},
  year      = {2024},
  publisher = {Schloss Dagstuhl -- Leibniz-Zentrum f{\"u}r Informatik},
  doi       = {10.4230/LIPIcs.ICDT.2024.13}
}

@misc{FullVersion,
  author    = {Simon Frisk and Paraschos Koutris},
  title     = {Conjunctive Queries with Negation: Beyond Signed-Acyclicity},
  year      = {2026},
  eprint    = {2610.03982},
  archivePrefix = {arXiv},
  url       = {https://arxiv.org/abs/2610.03982}
}

@article{CapelliCarmeliIrwinSalvati26,
  author    = {Florent Capelli and Nofar Carmeli and Oliver Irwin and Sylvain Salvati},
  title     = {Direct Access for Conjunctive Queries with Negations},
  journal   = {Logical Methods in Computer Science},
  volume    = {22},
  number    = {1},
  pages     = {24:1--24:53},
  year      = {2026},
  doi       = {10.46298/LMCS-22(1:24)2026}
}

@inproceedings{BringmannCarmeliMengel22,
  author    = {Karl Bringmann and Nofar Carmeli and Stefan Mengel},
  title     = {Tight Fine-Grained Bounds for Direct Access on Join Queries},
  booktitle = {Proceedings of the 41st ACM SIGMOD-SIGACT-SIGAI Symposium on Principles of Database Systems ({PODS})},
  pages     = {427--436},
  year      = {2022}
}

\newpage
\appendix
\section{Missing Proofs from Section~\ref{sec:nice-td}}
\label{sec:app:missing-proofs-sec3}

\bagscontaineddifferentslemma*
\begin{proof}
    For $j \in \{1,2\}$, let $\mH_0^{S_j}, \mH_1^{S_j}, \dots, \mH_n^{S_j}$ denote the sequence of hypergraphs obtained by running the canonical construction on $\mH^{S_j}$ with variable order $\sigma$, and let $U_i^{S_j} := \bigcup_{K \in \partial_{\mH_{i-1}^{S_j}} \sigma_i} K$ be the corresponding bag, so that $\chi_\sigma^{S_j}(\sigma_i) = U_i^{S_j}$. It therefore suffices to show that $U_i^{S_1} \subseteq U_i^{S_2}$ for every $i \in [n]$.

    We prove by induction on $i \in \{0, 1, \dots, n\}$ the following stronger claim: for every hyperedge $K_1 \in \mE(\mH_i^{S_1})$, there exists a hyperedge $K_2 \in \mE(\mH_i^{S_2})$ with $K_1 \subseteq K_2$.

    \emph{Base case ($i=0$).} Since $\mE(\mH_0^{S_1}) = \mE^+ \cup S_1 \subseteq \mE^+ \cup S_2 = \mE(\mH_0^{S_2})$, we may take $K_2 = K_1$.

    \emph{Inductive step.} Assume the claim holds for $i-1$. We first show $U_i^{S_1} \subseteq U_i^{S_2}$. Let $K_1 \in \partial_{\mH_{i-1}^{S_1}} \sigma_i$. By the inductive hypothesis, there exists $K_2 \in \mE(\mH_{i-1}^{S_2})$ with $K_1 \subseteq K_2$. Since $\sigma_i \in K_1 \subseteq K_2$, we have $K_2 \in \partial_{\mH_{i-1}^{S_2}} \sigma_i$, and therefore $K_1 \subseteq K_2 \subseteq U_i^{S_2}$. Taking the union over all such $K_1$ yields $U_i^{S_1} \subseteq U_i^{S_2}$.

    It remains to extend the inductive claim from $i-1$ to $i$. By the construction, $\mE(\mH_i^{S_j})$ consists of the edges of $\mH_{i-1}^{S_j}$ that do not contain $\sigma_i$, together with the new hyperedge $U_i^{S_j} \setminus \{\sigma_i\}$. The new edge is handled by the inclusion just established: $U_i^{S_1} \setminus \{\sigma_i\} \subseteq U_i^{S_2} \setminus \{\sigma_i\}$. For any inherited edge $K_1 \in \mE(\mH_{i-1}^{S_1})$ with $\sigma_i \notin K_1$, the inductive hypothesis gives $K_2 \in \mE(\mH_{i-1}^{S_2})$ with $K_1 \subseteq K_2$. If $\sigma_i \notin K_2$, then $K_2 \in \mE(\mH_i^{S_2})$ and we are done. Otherwise $K_2 \in \partial_{\mH_{i-1}^{S_2}} \sigma_i$, so $K_2 \subseteq U_i^{S_2}$, and hence $K_1 \subseteq K_2 \setminus \{\sigma_i\} \subseteq U_i^{S_2} \setminus \{\sigma_i\} \in \mE(\mH_i^{S_2})$.
\end{proof}

\sfhwequalsfhwlemma*
\begin{proof}
This is an immediate consequence of the definition and Lemma~\ref{lemma:TransformToCanonicalTD}. Indeed, for any $\sigma$ and $S=\emptyset$, we have $\chi_\sigma^S(\sigma_i) = \chi_\sigma(\sigma_i)$, where $\chi_\sigma(\sigma_i)$ is the bag where $\sigma_i$ disappears in the $\sigma$-canonical tree decomposition of $\smH$.
\end{proof}

\signedacyclicwidthonelemma*
\begin{proof}
    ($\Rightarrow$) Suppose $\smH$ is signed-acyclic. By the standard characterization of signed-acyclicity in terms of variable orders, there exists a variable order $\sigma$ such that for every $S \subseteq \mE^-$ and every $i \in [n]$, the bag $\chi_\sigma^S(\sigma_i)$ of the $\sigma$-canonical tree decomposition $\mT_\sigma^S$ of $\mH^S$ is contained in some hyperedge of $\mH^S$. For such a $\sigma$, $\rho_{\mH^S}^*(\chi_\sigma^S(\sigma_i)) = 1$ for every $S$ and every $i$, so $\sfhow(\smH) \le 1$. Since $\rho_{\mH^S}^* \ge 1$ on every non-empty set of vertices, we conclude that $\sfhow(\smH) = 1$.

    ($\Leftarrow$) Conversely, suppose $\sfhow(\smH) = 1$, witnessed by some variable order $\sigma$. Then for every $S \subseteq \mE^-$ and every $i \in [n]$, $\rho_{\mH^S}^*(\chi_\sigma^S(\sigma_i)) = 1$, which means that each bag $\chi_\sigma^S(\sigma_i)$ is contained in some hyperedge of $\mH^S$. Hence $\mT_\sigma^S$ witnesses $\textsf{fhw}(\mH^S) = 1$, so $\mH^S$ is $\alpha$-acyclic. Since this holds for every $S \subseteq \mE^-$, $\smH$ is signed-acyclic.
\end{proof}

\betafhwlesfhwlemma*
\begin{proof}
Let us write $f(\sigma, S) := \max_i \rho^*_{\mH^S}(\chi_\sigma^S(\sigma_i))$, so that
\[
\sfhow(\smH) \;=\; \min_{\sigma} \max_{S \subseteq \mE^-} f(\sigma, S),
\qquad
\beta\textsf{-fhw}(\smH) \;=\; \max_{S \subseteq \mE^-} \min_{\sigma} f(\sigma, S).
\]
The inequality is the standard max-min inequality applied to $f$. Concretely, let $\sigma^*$ be a variable order achieving the outer minimum in $\sfhow(\smH)$. Then for every $S \subseteq \mE^-$,
\[
\min_{\sigma} f(\sigma, S) \;\leq\; f(\sigma^*, S) \;\leq\; \max_{S' \subseteq \mE^-} f(\sigma^*, S') \;=\; \sfhow(\smH),
\]
where the first inequality holds because $\sigma^*$ is one specific choice of $\sigma$ and the second because $S$ is one specific choice of $S'$. Taking the maximum over $S$ on the left-hand side yields $\beta\textsf{-fhw}(\smH) \leq \sfhow(\smH)$.
\end{proof}

\iffullmode

\section{A Query Decomposition Lemma}

The lemma below is used in a similar way to the standard query decomposition lemma from previous work, arguing that we are bounded by the AGM bound when expanding to a new variable.

\begin{lemma}
\label{lemma:lightestextension}
    Let $\mH=(V,E)$ be a hypergraph. Let $U\subset V$ where $|U|=|V|-1$, and $X$ the variable in $V \setminus U$. Let $L$ be the set of outputs over the induced query of $\mH$ on $U$. Let $E_X$ be the hyperedges $e$ with $X\in e$. Let $\rho$ be a fractional edge cover of $\mH$. Then:
    \[
    \sum_{t\in L} \min_{R\in E_X} \deg(X | t) \leq O(|I|^{\rho})
    \]
\end{lemma}
\begin{proof}
    Suppose this is false. Suppose for each $t\in L$, we find extensions to $X$ using the atom $R$ reaching $X$ with the minimal degree, according to the above. Let $T$ be the set of all such obtained extensions. We claim there is an input instance $I$ to the query over $\mH$ such that the output of the query on $I$ is $T$. We construct $I$ as follows: for each $R\in E$, with $X\not\in R$, we use the same relation instance. If $X\in R$, for each $\pi_{R\setminus\{X\}}$, letting $d=\deg_R(X|R\setminus\{X\})$, we create $d$ tuples in $R$, each with values $\pi_{R\setminus\{X\}}$ on $R\setminus\{X\}$ and values $1,\dots, d$ on $X$. It is not hard to see that the output size is now $|T|$, since for each $t\in L$, the number of extensions to $X$ is exactly $\min_{R\in E_X} \deg(X | t)$. But then, the output size is above the AGM bound, which is a contradiction.
\end{proof}

\section{Induced Subquery Lemma}
\label{sec:appendix-induced-subquery}

\begin{lemma}
\label{lemma:InducedSubquery}
    Let $\mH=(V,E)$ be a hypergraph and let $U\subseteq V$. Then
    $
    \textsf{AGM}(\mH[U]) \leq \textsf{AGM}(\mH).
    $
\end{lemma}
\begin{proof}
    Let $\mathbf{x}=(x_e)_{e\in E}$ be an optimal fractional edge cover for $\mH$. Restricting the linear program constraints to vertices in $U$ yields a feasible fractional edge cover for $\mH[U]$ with the same objective value, because for each $u\in U$ we still have $\sum_{e\ni u} x_e \ge 1$. Therefore the optimum for $\mH[U]$ cannot exceed the optimum for $\mH$, i.e., $\textsf{AGM}(\mH[U]) \leq \textsf{AGM}(\mH)$.
\end{proof}

\section{Algorithm Analysis}
\label{sec:app:algorithm-analysis}

In this section, we give the full analysis of the algorithm, resulting in the proof of the main theorem.

\subsection{Runtime}

\begin{lemma}
    $\textproc{Extend}(v,S)$ runs in $O(1)$.
\end{lemma}
\begin{proof}
    Recall that each node in the linked list has a hash table storing skip links which start/end at that node, where the key is the label of the skip link and the value is the start and end nodes.

    $\textproc{Extend}(v,S)$ can therefore naively find consistent skip links by checking each subset $S'\subset S$. For each $S'$, there is at most one incoming/outgoing skip link with a consistent label. It therefore runs in constant time.
\end{proof}

\begin{lemma}
    $\textproc{Merge}(T)$ runs in $O(|T| + |L^\emptyset|)$.
\end{lemma}
\begin{proof}
    Implement $\textproc{Merge}(T)$ as a linear scan of each linked list. During each scan, we maintain a hash set of skip links $s$ for which the current visited node $v$ is between $s_{start}$ and $s_{end}$, where the hash key is $s_{label}$. When visiting $v$, we check if any skip link $s'$ which starts at $v$ has the same label as a skip link $s$ in the hash set. If so, we remove $s$ and $s'$ and add in a merged skip link $s''=(s_{start}, \max\{s_{end}, s'_{end}\}, s_{label})$.

    We pay the cost of $O(|T|)$ since each skip link creates at most the following work: one addition to the hash set, one removal from the hash set, one merge check and one merge operation, each of which is $O(1)$. Additionally, we need to pay the cost of iterating over all linked lists, which has cost $O(|L^\emptyset|)$ since this is the total number of nodes across all linked lists.
\end{proof}

\begin{lemma}
\label{lemma:canuseinputstats}
    Denote by $\textsf{AGM}(H)$ the AGM bound on hypergraph $H$.
    For any $i$, for any $S\subseteq \mE_i^-$, for any $X \subseteq V_i$, the join over $\mH_i^S[X]$ is bounded by $\textsf{AGM}(\mH^{\mA_S}[X])$.
\end{lemma}
\begin{proof}
    We will prove the following with induction over $i$: for any $R' \in \mE^+ \cup \mA_S$, there is a $R \in \mE_i^+\cup S$ such that (1) $V_i \cap R' \subseteq R$ (sets of variables) and (2) $\pi_{V_i \cap R'}R \subseteq \pi_{V_i \cap R'}R'$ (sets of tuples), where $R \in S$ if $R'$ is negative. Hence we can transfer the weight of $R'$ in a fractional edge cover of $\mH^{\mA_S}[X]$ to a cardinality constraint over $\pi_{V_i \cap R'}R$ over the variables $V_i \cap R'$, which gives the same AGM bound for $\mH_i^S[X]$. For $i=1$, the statement is trivial since $\smH_1=\smH$ and $\mA_S=S$, so we may take $R=R'$.

    For the inductive step, let
    \[
    S_{i-1} = (S \cap \mE_{i-1}^-) \cup \bigcup_{U_{i-1}^{S'}\setminus\{\sigma_{i-1}\} \in S} S'
    \]
    be the set of atoms at step $i-1$ that $S$ is built from. Since $\mA_{U_{i-1}^{S'}} = \mA_{S'}$, we have $\mA_{S_{i-1}} = \mA_S$, so by the induction hypothesis applied to $S_{i-1}$, there is $R_{i-1} \in \mE_{i-1}^+ \cup S_{i-1}$ satisfying (1) and (2), where $R_{i-1} \in S_{i-1}$ if $R'$ is negative. If $\sigma_{i-1} \notin R_{i-1}$, then $R_{i-1}$ is carried over to step $i$, and we take $R = R_{i-1}$. Otherwise, $R_{i-1}$ is used to create a new atom, which we take as $R$: if $R_{i-1}$ is positive, $R = U_{i-1}^\emptyset\setminus\{\sigma_{i-1}\}$, and if it is negative, $R_{i-1} \in S'$ for some $R = U_{i-1}^{S'}\setminus\{\sigma_{i-1}\} \in S$. Since $R$ is a projection of a join containing $R_{i-1}$, (1) and (2) carry over to $R$.
\end{proof}

\enumdsconstructionlemma*
\begin{proof}
    Let $I$ be the input size of the original query and let $I_P$ be an upper bound on $|U_j^S|$ over all $j<i$ and $S$. Denote by $T = |I|^{\sfhow(\smH)}$ the target cost. Since all $|U_j^S|$ are the output of line $2$ or $17$ at step $j$, showing those lines are bounded by $T$ inductively gives $I_P \leq T$; we thus treat $I_P \leq T$ as given in what follows.

    By Lemma~\ref{lemma:canuseinputstats}, every join over $\mH_i^S[X]$ below is bounded by $\textsf{AGM}(\mH^{\mA_S}[X]) \leq T$, where the inequality follows from the definition of $\sfhow(\smH)$. To prove the lemma, it suffices to bound the cost of each line by $\textsf{AGM}(\mH^{\mA_S}[X]) + |I| + |I_P|$.

    We claim we need to bound the cost of the following lines by $T$: $2$, $7$, $11$ and $15$. By the previous lemmas on the runtime of $\textproc{Extend}$ and $\textproc{Merge}$, other lines run in time linear of the output size of the mentioned lines.

    \textbf{Lines 2 and 7}: by Lemma~\ref{lemma:canuseinputstats}, these run in $O(|I| + |I_P| + \textsf{AGM}(\mH^\emptyset[\chi_\sigma^\emptyset(\sigma_i)]))$ and $O(|I| + |I_P| + \textsf{AGM}(\mH^{\mA_R}[\chi_\sigma^{\mA_R}(\sigma_i)]))$ respectively, both bounded by $T$.

    \textbf{Line 11}: by Lemma~\ref{lemma:canuseinputstats}, the runtime is $O(|I| + |I_P| + \textsf{AGM}(\mH^{\mA_S}[\chi_\sigma^{\mA_S}(\sigma_i) \setminus \{\sigma_i\}]))$. By Lemma~\ref{lemma:InducedSubquery}, the AGM term is at most $\textsf{AGM}(\mH^{\mA_S}[\chi_\sigma^{\mA_S}(\sigma_i)]) \leq T$.

    \textbf{Line 15}: the total cost across iterations of line $13$ is
    \[
        \sum_{t\in L^S} \min_{R\in S} \deg_{P^R}(\sigma_i \mid t).
    \]
    We claim for any atom $R' \in \mH_i^S[\chi_\sigma^{\mA_S}(\sigma_i)]$, $R'$ is in some $P^R$ (if it is positive, it is in all, if it is negative, it is in $P^R$ where $R=R'$).
    Then $\deg_{P^R}(\sigma_i\mid t) \leq \deg_{R'}(\sigma_i\mid t)$. By Lemma~\ref{lemma:lightestextension} applied to $\mH_i^S[\chi_\sigma^{\mA_S}(\sigma_i)]$ together with Lemma~\ref{lemma:canuseinputstats}, this sum is $O(|I| + |I_P| + \textsf{AGM}(\mH^{\mA_S}[\chi_\sigma^{\mA_S}(\sigma_i)])) \leq T$.
\end{proof}

\subsection{Semantic Lemmas}

We now present two lemmas which formalize precisely when a certain skip link exists in the data structure. This will help us analyze both correctness and efficiency.

\begin{lemma}
\label{lemma:SemanticLemmaForward}
    Let $S \subseteq \mC$.
    Let $\sigma_i$ be a variable in the variable order, $t_{i+1}$ a suffix tuple over $\smH_{i+1}$, $\mL$ a linked list in $\mathbb{D}_i$, and $A=\{a_1,\dots,a_m\}$ a set of consecutive values in $\mL$.
    For an atom $R\in S$, we say $R$ witnesses $a\in A$ if there is a tuple $t_R\in R$ with $\pi_{\sigma_i}(t_R)=a$ and $t_R \sim t$.
    If there is a skip link $s$ in $\mL$ with $\textsf{prov}(s) = S$ covering $A$, with label and bucket $t=\ell\cdot b$, then:
    \begin{enumerate}
        \item $t \in \textproc{Join}(\mH_i^{S}[\chi_\sigma^{\mA_{S}}(\sigma_i)\setminus\{\sigma_i\}])$
        \item For every $a\in A$, there is an $R\in S$ witnessing $a$
    \end{enumerate}
\end{lemma}
\begin{proof}
    We prove this by induction over $|S|$.

    \textbf{Base case}: $S=\{R\}$. (1) holds since otherwise $t\notin L^{\{R\}}$ and iteration $\{R\}$ would not create $s$. On lines~14--15, $R^*=R$ and $V^{\{R\}}=t\bowtie P^R$. Since no $S''\subsetneq\{R\}$ exists, $s_{in}=s_{out}=\textsf{null}$ in \textproc{Extend}, so each $v\in V^{\{R\}}$ yields the singleton $(\textsf{prev}_\mathbb{D}(\pi_{\sigma_i}v),\textsf{next}_\mathbb{D}(\pi_{\sigma_i}v),\ell)$. \textproc{Merge} only merges non-containingly overlapping singleton skip links and since $s$ covers $A$, every $a\in A$ must have such singleton skip link, giving (2).

    \textbf{Inductive Step}: $|S|>1$. Assume the lemma holds for all $S^\dagger\subsetneq S$. (1) holds by the same argument as in the base case. It remains to prove (2) for each $a\in A$. First, $s$ is output from \textproc{Merge}, which merged several outputs from \textproc{Extend}. The inputs to these \textproc{Extend} outputs are tuples in $V^S$. We can use the $R^*$ creating these tuples to witness a non-empty subset $A' \subseteq A$. Let $A_1,A_2,\dots$ be the consecutive segments of $A\setminus A'$ in $\mL$. Each $A_j$ must be covered by some skip link $s^j$ with $\textsf{prov}(s^j) \subsetneq S$, which was used by \textproc{Extend}, since otherwise, the resulting skip links would not have been merged to cover all of $A$. By the inductive hypothesis applied to $s^j$, every $a\in A_j$ has a witness in $S^j\subseteq S$.
\end{proof}

\begin{lemma}
\label{lemma:SemanticLemmaBackward}
    Let $\sigma_i$, $t_{i+1}$, $\mL$, $A$ be as above, and fix $S\subseteq\mC$.
    For an atom $R\in S$, we say $R$ witnesses $a\in A$ if there is a tuple $t_R\in R$ with $\pi_{\sigma_i}(t_R)=a$ and $t_R \sim t$.
    If the following hold:
    \begin{enumerate}
        \item $t \in \textproc{Join}(\mH_i^{S}[\chi_\sigma^{\mA_{S}}(\sigma_i)\setminus\{\sigma_i\}])$
        \item For every $a\in A$, there is an $R\in S$ witnessing $a$
        \item For every $R\in S$, there is an $a\in A$ witnessed by $R$ but not by any $R'\in S$ where $R'\neq R$
        \item The values bordering $A$ in $\mL$ ($\textsf{prev}_\mathbb{D}(\min_\mathbb{D} A)$ and $\textsf{next}_\mathbb{D}(\max_\mathbb{D} A)$, if they exist) are not witnessed by any $R\in S$.
    \end{enumerate}
    then there is a skip link $s$ in $\mL$ with $\textsf{prov}(s)=S$, label and bucket $t$, $s_{start}=\textsf{prev}_\mathbb{D}(\min_\mathbb{D} A)$ and $s_{end}=\textsf{next}_\mathbb{D}(\max_\mathbb{D} A)$
\end{lemma}
\begin{proof}
    We prove this by induction over $|S|$.

    \textbf{Base case}: $S=\{R\}$. By (1), $t \in L^{\{R\}}$. By (2), every $a\in A$ yields a singleton skip-link from $\textsf{prev}_\mathbb{D}(a)$ to $\textsf{next}_\mathbb{D}(a)$ in \textproc{Extend}. These are non-containingly overlapping since $A$ is consecutive, so \textproc{Merge} merges them into one skip link from $\textsf{prev}_\mathbb{D}(\min_\mathbb{D} A)$ to $\textsf{next}_\mathbb{D}(\max_\mathbb{D} A)$ with label $\ell$, tagged $\textsf{prov}\gets\{R\}$.

    \textbf{Inductive Step}: $|S|>1$. Assume the lemma holds for all $S^\dagger\subsetneq S$. By (1), $t \in L^S$. Let $R^*$ be the choice on line~14, and let $A'\subseteq A$ be the subset of $A$ such that $R^*$ uniquely witnesses each $a\in A'$. For each $a\in A'$, $V^S$ contains a tuple $t \cdot a$.

    Let $A_1,A_2,\dots$ be the consecutive segments of $A \setminus A'$ in $\mL$. Each segment $A_j$ satisfies (1)--(4) with respect to some $S^j \subseteq S\setminus\{R^*\}$. (1) holds since $t$ satisfies (1) for $S$. (2) holds since each $a\in A_j$ is witnessed by some $R\in S$ but not uniquely by $R^*$, so $S^j \subsetneq S$ suffices. (3) holds since if there is a $R'\in S_j$ where it does not hold, we can edit $S_j$ by removing $R'$, and doing this until (3) holds. (4) holds since each of $\textsf{prev}_\mathbb{D}(\min_\mathbb{D} A_j)$ and $\textsf{next}_\mathbb{D}(\max_\mathbb{D} A_j)$ is either a value in $A'$ (uniquely witnessed by $R^*\notin S^j$) or equals $\textsf{prev}_\mathbb{D}(\min_\mathbb{D} A)$ or $\textsf{next}_\mathbb{D}(\max_\mathbb{D} A)$ respectively (not witnessed by any $R\in S\supseteq S^j$ by (4) for $A$).

    By the inductive hypothesis, there is a skip link $s^j$ with $\textsf{prov}(s^j)=S^j\subsetneq S$, label consistent with $t$, $s^j_{start}=\textsf{prev}_\mathbb{D}(\min_\mathbb{D} A_j)$ and $s^j_{end}=\textsf{next}_\mathbb{D}(\max_\mathbb{D} A_j)$. These endpoints are either in $A'$ or at the border of $A$.

    Therefore, for each $a' \in A'$, \textproc{Extend} finds incoming/outgoing skip links reaching precisely to the next nodes in $A'$ (or to the borders of $A$), producing skip links over $A$ that are non-containingly overlapping. \textproc{Merge} then combines them into the desired $s$ with $s_{start}=\textsf{prev}_\mathbb{D}(\min_\mathbb{D} A)$ and $s_{end}=\textsf{next}_\mathbb{D}(\max_\mathbb{D} A)$, tagged with provenance $S$.
\end{proof}

We remark that there may be skip links $s$ and $s'$ which are non-containingly overlapping, and where $s_{label}$ and $s'_{label}$ are consistent, but there is no skip link $s''$ merged from them. This happens when there is no $t_{i+1}$ which is consistent with \textit{both} of $s$ and $s'$.

\subsection{Correctness and Constant Delay}

Algorithm~\ref{algorithm:BuildEnumDS} constitutes a reduction from $\smH_i$ to $\smH_{i+1}$, in that we compute $\smH_i$ by computing $\smH_{i+1}$, with an input created by the reduction, and then iterate over extensions to the next variable $\sigma_i$. We argue that this reduction is correct.

We will say that a tuple $t$ satisfies an atom $R$, if either $R$ is positive and $t$ is consistent with a tuple in $R$, or $R$ is negative and $t$ is not consistent with any tuple in $R$.

\reductioncorrectlemma*
\begin{proof}
    We perform induction on $i$. The base case $i=1$ is immediate, since $Q_1(I_1)=Q(I)$ and $V_1=V$. For the inductive step, assuming $Q_i(I_i)=\pi_{V_i}Q(I)$, it suffices to show:
    \[
        Q_{i+1}(I_{i+1}) = \pi_{V_{i+1}} Q_i(I_i).
    \]
    Let $t_i$ be a tuple over $V_i$, and $t_{i+1}=\pi_{V_{i+1}}(t_i)$ be the projection of $t_i$ to $V_{i+1}$. We assume that $t_i$ and $t_{i+1}$ satisfy all atoms that occur in both $\smH_i$ and $\smH_{i+1}$, since otherwise neither is an output (these atoms do not contain $\sigma_i$, so $t_i$ satisfies such an atom iff $t_{i+1}$ does). It therefore suffices to show:
    \begin{enumerate}
        \item If $t_i$ satisfies all atoms in $\smH_i$ containing $\sigma_i$, then for each $S$, $t_{i+1}$ satisfies atoms $U_i^S$.
        \item If $t_{i+1}$ satisfies all atoms $U_i^S$, then there is some $t_i$ where $t_i \sim t_{i+1}$ which satisfies all atoms in $\smH_i$ containing $\sigma_i$.
    \end{enumerate}

    \textbf{(1)}:
    We prove the contrapositive. Suppose there is some $S$ where $t_{i+1}$ does not satisfy $U_i^S$. If $S=\emptyset$, then $U_i^\emptyset$ is the positive atom $\textproc{Join}(\mH_i^\emptyset[\chi_\sigma^\emptyset(\sigma_i)])$. $t_{i+1}$ not satisfying it means every $t_i\sim t_{i+1}$ does not satisfy some positive atom containing $\sigma_i$. If $S\neq\emptyset$, then $U_i^S$ is a negative atom containing skip-link spanning an entire linked lists. $t_{i+1}$ not satisfying it means $t_{i+1}$ is consistent with the label/bucket of some such skip link $s$. By Lemma~\ref{lemma:SemanticLemmaForward}~(2) applied to $s$ with $A=\mL$, every $x\in\mL$ is witnessed by some $R\in S$, so for every $t_i=x\cdot t_{i+1}$ where $x\in\mL$, some negative $R\in S$ is not satisfied.

    \textbf{(2)}:
    We again prove the contrapositive. Suppose there is no $t_i$ where $t_i \sim t_{i+1}$ satisfies all atoms $R$ in $\smH_i$ containing $\sigma_i$. We consider two cases.

    First, suppose for every $t_i$ there is a positive $R$ which is not satisfied. Then $U_i^\emptyset=\textproc{Join}(\mH_i^\emptyset[\chi_\sigma^\emptyset(\sigma_i)])$ has no tuple consistent with $t_{i+1}$, so $t_{i+1}$ does not satisfy $U_i^\emptyset$.

    Second, let $X$ be the values over $\sigma_i$ such that $t_i=x\cdot t_{i+1}$ satisfies all positive $R$ with $\sigma_i\in R$. Note $X$ is exactly the values appearing in the linked list $\mL$ in bucket $B$ consistent with $t_{i+1}$, since $\mL$ is built from the join of the positive atoms. Suppose for every $x\in X$ the corresponding $t_i$ has a negative atom $R$ with $\sigma_i\in R$ which is not satisfied, and let $S$ be a minimal such set. By Lemma~\ref{lemma:SemanticLemmaBackward} applied with $A=X=\mL$, there is a skip link $s$ with $\textsf{prov}(s)=S$ spanning the entire linked list with label consistent with $t_{i+1}$. This link is passed up in $U_i^S$, which matches $t_{i+1}$, so $t_{i+1}$ would not be output over $\smH_{i+1}$. It remains to verify the conditions of Lemma~\ref{lemma:SemanticLemmaBackward}.

    (1): Consider the query over $\mH_i^S[\chi_\sigma^{\mA_S}(\sigma_i)\setminus\{\sigma_i\}]$. Atoms also in $\smH_{i+1}$ are satisfied by $t_{i+1}$'s witnesses. The remaining atoms are those of $S$. Each $R\in S$ contributes the tuple $t_R\sim t_{i+1}$ that witnesses some $x\in X$.

    (2): For each $x\in X$, we have established some such negative $R$ exists.

    (3): This is true by the minimality of $S$.

    (4): This holds since $A=X$ is the entire linked list.
\end{proof}

Next, we show that the \textproc{Enumerate} function of the enumeration data structure correctly enumerates extensions to the next variable $\sigma_i$.

\enumfunctionstepoutputslemma*
\begin{proof}
    First, the query over $\mH_i^\emptyset[\chi_\sigma^\emptyset(\sigma_i)]$, which gives the hash table, will have precisely extensions to $t_{i+1}$ which satisfy the positive atoms $R$ with $\sigma_i\in R$ (other $R$ are satisfied by the same witnesses as $t_{i+1}$). It remains to show the following: in the linked list $\mL$, let $\mL'$ be the (not necessarily contiguous) subsequence of $\mL$ which are the correct extensions, i.e. the values over $\sigma_i$ which are not blocked by any negative atom. We need to argue that when we follow skip links, we traverse exactly $\mL'$.

    We first argue that no skip link consistent with $t_{i+1}$ can skip a value in $\mL'$. Indeed, by Lemma~\ref{lemma:SemanticLemmaForward}~(2), if there were a skip link consistent with $t_{i+1}$ that skipped a value in $\mL'$, then there would be a negative atom that witnesses this skip, which contradicts the definition of $\mL'$.

    Next, we argue that for any pair of values $x_1,x_2 \in \mL'$, which are consecutive in $\mL'$ but not in $\mL$ (where $x_1$ may also be the bucket key and $x_2$ the sentinel), there is a skip link $s$ with label consistent with $t_{i+1}$ with $s_{start}=x_1$ and $s_{end}=x_2$.

    Let $A$ be the values in $\mL$ strictly between $x_1$ and $x_2$. Apply Lemma~\ref{lemma:SemanticLemmaBackward}: (2) holds since every $a\in A$ is blocked; let $S^*$ be a minimal such set, giving (3); (1) holds by using the same witnesses as for $t_{i+1}$ and the tuples in $S^*$ which witness each $a\in A$; (4) holds since $x_1,x_2\in\mL'$ are not witnessed by any negative atom. The resulting $s$ has $s_{start}=x_1$ and $s_{end}=x_2$.

    For constant delay, at each node, the algorithm inspects the at most $2^{|\mE_i^-|}$ outgoing skip links with label $\sim t_{i+1}$ and follows the longest. By the arguments above, the longest one from $x_1\in\mL'$ has $s_{end}=x_2$, so following one pointer suffices between outputs.
\end{proof}

\subsection{Proving the Main Theorem}

We now put together the pieces to prove the main theorem.

\mainthm*
\begin{proof}\mbox{}
    \textbf{Correctness}:
    By Lemma~\ref{lemma:ReductionCorrect}, the set of outputs of each reduced query satisfies $Q_i(I_i) = \pi_{V_i}Q(I)$. We argue inductively, starting from $i=n$, that the algorithm correctly enumerates $Q_i(I_i)$. The base case is $\smH_{n+1}$, the empty hypergraph, whose unique output is the empty tuple $t_{n+1}=()$. For each $i$ from $n$ down to $1$, Lemma~\ref{lemma:enumfunctionstepoutputs} applied to (any output) $t_{i+1}\in Q_{i+1}(I_{i+1})$ guarantees that \textproc{Enumerate} produces exactly the outputs $t_i\sim t_{i+1}$ over $\smH_i$. Taken over all $t_{i+1}$, this yields all of $Q_i(I_i)$. At $i=1$ we obtain $Q(I)$.

    \textbf{Preprocessing Time}: The preprocessing phase consists of $n$ calls to \textproc{BuildEnumDS}, which by Lemma~\ref{lemma:EnumDSConstruction} runs in $O(|I|^{\sfhow(\smH)})$.

    \textbf{Constant Delay}:
    Lemma~\ref{lemma:enumfunctionstepoutputs} shows that each step \textproc{Enumerate} in one attribute has constant delay. It remains to argue we only need to call this a constant number of times between different calls to \textproc{Emit}. By Lemma~\ref{lemma:ReductionCorrect}, every $t_i \in Q_i(I_i)$ is the projection of some $t \in Q(I)$, so each enumerated $t_i$ extends to at least one output of the original query. Hence, we need to do at most $2n$ calls to \textproc{Enumerate}, where $n$ is the number of variables (each \textproc{Enumerate} may run to the end of its linked list, causing us to backtrack, and then be called again to enumerate the next value).
\end{proof}

\section{Proof for $k$-Cycle Queries}
\label{sec:app:proofkcycle}

\kcycletheorem*
\begin{proof}
    Consider a query in the family $C_k^{\pm}$ with at least one negative edge and an instance $I$ with cardinality at most $m$. Let $R_k$ be a negative edge. Partition $R_k$ into two subrelations $R_k^l$ and $R_k^h$ such that $\deg_{R_k^l}(x_k \mid x_1) \leq \sqrt{m}$ and $\deg_{R_k^h}(x_1 \mid x_k) \leq \sqrt{m}$. This gives two subinstances of $I$: $I^l$ and $I^h$, where we also restrict $X_1$ to the values of $x_1$ with degree at most $\sqrt{m}$ in $R_k$ for $I^l$, and to the remaining values for $I^h$. For $I^l$, pick the variable order $x_1, \dots, x_k$. For $I^h$, pick the variable order $x_k, \dots, x_1$. Since $I^l$ and $I^h$ partition the values of $x_1$, and for every value of $x_1$ the atom $\neg R_k$ is equivalent to $\neg R_k^l$ in $I^l$ (resp.\ $\neg R_k^h$ in $I^h$), we enumerate $Q(I)$ by enumerating $Q(I^l)$ and then $Q(I^h)$. W.l.o.g. consider $I^l$ for the remainder of the proof.

    We first perform an intermediate preprocessing step on each subinstance, where we perform semijoins/antijoins to enforce condition (2) of Lemma~\ref{lemma:keykcyclelemma}. This is done in linear time.

    We will use the regular $\textproc{BuildEnumDS}$ procedure, with a small modification we describe in the next paragraph. In step $i$, when we remove variable $x_i$, there is a negative atom $R=\{x_i, x_k\}$ with degree constraint $\deg_R(x_k|x_i)\leq\sqrt m$ by Lemma~\ref{lemma:keykcyclelemma} applied to the path $x_k\to x_1\to\cdots\to x_i$ ((1) holds by the partition, (2) holds by the preprocessing). We therefore get: $|\mH_i^S[\chi_\sigma^{\mA_S}(x_i)]|\leq m\sqrt m$, which translates to the runtime of line~7 of \textproc{BuildEnumDS}.

    For lines 10--17, the bound becomes more subtle, as the induced subquery lemma, which we used in the original proof, does not apply with degree constraints. We will slightly modify how this step is done, considering different cases on the bag $\chi_\sigma^{\mA_S}(x_i)$:
    \begin{itemize}
        \item $\chi_\sigma^{\mA_S}(x_i)=\{x_i\}$: We use the normal method, and the runtime is bounded by the input size.
        \item $\chi_\sigma^{\mA_S}(x_i)\subseteq\{x_i,x_{i+1}\}$: We use the normal method, and the runtime is bounded by the input size.
        \item $\chi_\sigma^{\mA_S}(x_i)=\{x_i,x_k\}$: We use the normal method. The runtime is bounded by $m\sqrt{m}$, which is the size of $U_i^S$.
        \item $\chi_\sigma^{\mA_S}(x_i)=\{x_i,x_{i+1},x_k\}$: We will modify how we compute $L^S$, since line~11 is a Cartesian product with cost $m^2$. We first compute an intermediate $W$ as follows: If $R_i=(x_i,x_{i+1})$ is positive, $W=R\bowtie R_i$, and if $R_i=(x_i,x_{i+1})$ is negative, $W$ joins $R$ with $R_i$ by matching a tuple $t\in R$ with $t'\in R_i$ if $t.x_i$ and $t'.x_i$ are consecutive in the linked list for $x_i$ (which has only $1$ bucket). $|W| \leq m\sqrt{m}$ by the degree constraint on $R$. We then compute $L^S$ by projecting $W$ on $x_k$ and $x_{i+1}$, and doing any semijoin on unary atoms. Finally, line~15 is bounded by $|W|$ by first performing a semijoin of $P^{R^*}$ with $W$.
    \end{itemize}

    \textbf{Correctness}: The modification means that $L^S \subseteq \hat{L}^S$, where $\hat{L}^S$ is the set of labels and buckets obtained in the original procedure. Consider $t = \ell \cdot B \in \hat{L}^S \setminus L^S$. Then $t \in \textproc{Join}(\mH_i^S[\chi_\sigma^{\mA_S}(\sigma_i)\setminus\{\sigma_i\}])$. If $R_i$ is positive, then $t$ will be filtered out by the join with $P^{R^*}$, since $P^R$ contains both atoms in $W$. It remains the case when $R_i$ is negative.

    Since $t$ is in $\hat{L}^S \setminus L^S$, it means that considering the set $A$ of nodes in the linked list $\mL$ ($t$ determines a bucket) which are skipped by $R_{i}$ or $R$. Since $t$ is not in $L^S$, it must be that there are no two nodes $a_1,a_2 \in A$ which are consecutive in the linked list, where $a_1$ is only skipped by $R_{i}$ and $a_2$ is only skipped by $R$. Let $A_1,A_2,\dots$ be the consecutive subsequences of $A$. Then, for each $A_j$, we can pick either $R$ or $R_i$ to block the values in $A_j$, consistent with $t$. We can apply Lemma~\ref{lemma:SemanticLemmaBackward} to show that there is a skip link over each such $A_j$ already, with provenance $\{R\}$ or $\{R_i\}$. This skip link has the same start and end as the one that would have been created from $t$, and the label is a subset of the original label. Therefore, leaving out the skip links from $t$ does not affect the algorithm, since we can use these other skip links instead.

    For the modification of line~15, where we first semijoin $P^{R^*}$ with $W$, we note that a tuple $t$ filtered out by this semijoin is also in such an interval $A_j$, and by the same argument as above, these tuples are already covered by other skip links.

    \textbf{Constant Delay}: The constant delay argument is unaffected by the modifications, which only affect that some skip links are removed compared to the unchanged algorithm. By the same argument as above, if the algorithm was to follow such a skip link, there is another skip link with the same start and end, with a label which is a subset of the original label. If constant delay is achieved in the original algorithm, it is also achieved in the modified algorithm.
\end{proof}

\fi

\end{document}